\documentclass[runningheads]{llncs}
\usepackage[T1]{fontenc}
\usepackage{graphicx}
\usepackage{caption}
\usepackage{subcaption}
\usepackage{thmtools} 
\usepackage{thm-restate}
\usepackage{amsmath,amssymb}

\usepackage{hyperref}
\usepackage{color}

\graphicspath{{./img/}{./lipics-authors-v2021.1.3/}}
\usepackage{multirow}
\usepackage{colortbl}
\usepackage{svg}
\PassOptionsToPackage{table,xcdraw}{xcolor}
\usepackage{apptools}
\newcommand{\restateref}[1]{\IfAppendix{\hyperref[#1]{$\star$}}{\hyperref[#1*]{$\star$}}}
\newcommand{\oh}{\ensuremath{\mathcal{O}}}
\newcommand{\p}{\ensuremath{\mathcal{P}}}
\newcommand{\tp}{\ensuremath{\mathcal{\widetilde P}}}
\usepackage{paralist}
\usepackage{enumitem}
\usepackage{booktabs}

\usepackage{dsfont}

\spnewtheorem{myclaim}{Claim}{\bfseries}{\itshape}
\renewcommand{\qed}{\hfill$\Box$}

\usepackage{cleveref}
\usepackage{orcidlink}

\crefname{figure}{Fig.}{Figs.}
\crefname{theorem}{Thm.}{Thms.}
\crefname{lemma}{Lem.}{Lems.}
\crefname{corollary}{Cor.}{Cors.}
\crefname{observation}{Obs.}{Obs.}
\crefname{section}{Sec.}{Secs.}

\begin{document}
	\title{Hypergraphs as Metro Maps: Drawing Paths with Few Bends in Trees, Cacti, and Plane 4-Graphs}
	\titlerunning{Visualizing Hypergraphs as Metro Maps}
	\author{Sabine Cornelsen\inst{1}\orcidlink{0000-0002-1688-394X}
		\and
		Henry F\"orster\inst{2}\orcidlink{0000-0002-1441-4189} \and
		Siddharth Gupta\inst{3}\orcidlink{0000-0003-4671-9822} \and
		Stephen Kobourov\inst{2}\orcidlink{0000-0002-0477-2724}
		\and
		Johannes Zink\inst{2}\orcidlink{0000-0002-7398-718X}
	}
	\authorrunning{S. Cornelsen et al.}
	\institute{University of Konstanz, Germany
		\and
		TU Munich, Heilbronn, Germany
		\and
		BITS Pilani, K K Birla Goa Campus, India
		}
	\maketitle

\begin{abstract}
A \emph{hypergraph} consists of a set of vertices and a set of subsets of vertices, called \emph{hyperedges}.
In the metro map metaphor, each hyperedge is represented by a path (the metro line) and the union of all these paths is the support graph (metro network) of the hypergraph. 
Formally speaking, a \emph{path-based support} is a graph together with a set of paths. We consider the problem of 
constructing 
drawings of path-based 
supports that
\begin{inparaenum}[(i)]
\item minimize the 
sum of the number of bends on all paths, 
\item minimize the
maximum number of bends on any path, or
\item maximize the number of 0-bend paths, then the number of 1-bend paths, etc. 
\end{inparaenum}
We concentrate on straight-line 
drawings of path-based tree and cactus supports as well as orthogonal drawings of path-based plane supports with maximum degree~4.

\end{abstract}

\section{Introduction}

A \emph{hypergraph} $H=(V,A)$ is a set $V$ of vertices and a set $A$ of subsets of $V$, called hyperedges. Among others,
hypergraphs can be visualized as
bipartite graphs 
with vertices $V \cup A$, as \emph{Hasse diagrams} (upward drawing of the transitive reduction of the inclusion relation between hyperedges), as Euler diagrams (enclosing the vertices of each hyperedge by a simple closed curve), or as incidence matrices. For an overview on more visualization styles,
also see~\cite{alsallakh2016state,fischer:survey21}.
We consider the metro map metaphor, where each hyperedge~$h$ is visualized with a metro line along which the vertices in the hyperedge are the stations; see \cref{fig:metrosets}
for an examples of a hypergraph made by Simpson characters
and an example of a hypergraph constructed from the authors of this paper and its references.
That is, each metro line is a simple path $p_h$ with the same vertices as~$h$, the union of the paths 
is a \emph{path-based support} of the hypergraph
$H=(V,A)$: a \emph{support graph} $G=(V,E)$ and a set $\p=\{p_h;h\in A\}$ of paths of $G$ such that
each edge of $G$ is contained in at least one path of $\p$.
A \emph{cactus} is a connected graph where each 2-connected component is an edge (bridge) or a simple cycle. A \emph{plane 4-graph} is a
planar graph of maximum degree~4 with a fixed planar embedding.\footnote{%
    We use \emph{embedding} in the combinatorial sense: it prescribes
    at each vertex the order of incident edges (rotation system) and defines the outer face.
    We do not assume a given rotation system for trees and~cacti.}
We say that $G$ is a \emph{path-based planar support}, a \emph{path-based tree support}, a \emph{path-based cactus support}, or a \emph{path-based plane 4-graph support},  if $G$ is a planar graph, a tree, a cactus,
or a plane 4-graph, respectively.

\begin{figure}[t]
    \begin{subfigure}[b]{.35\textwidth}
        \centering
        \includegraphics[width=\linewidth,trim={2cm 0cm 26cm 1cm},clip]{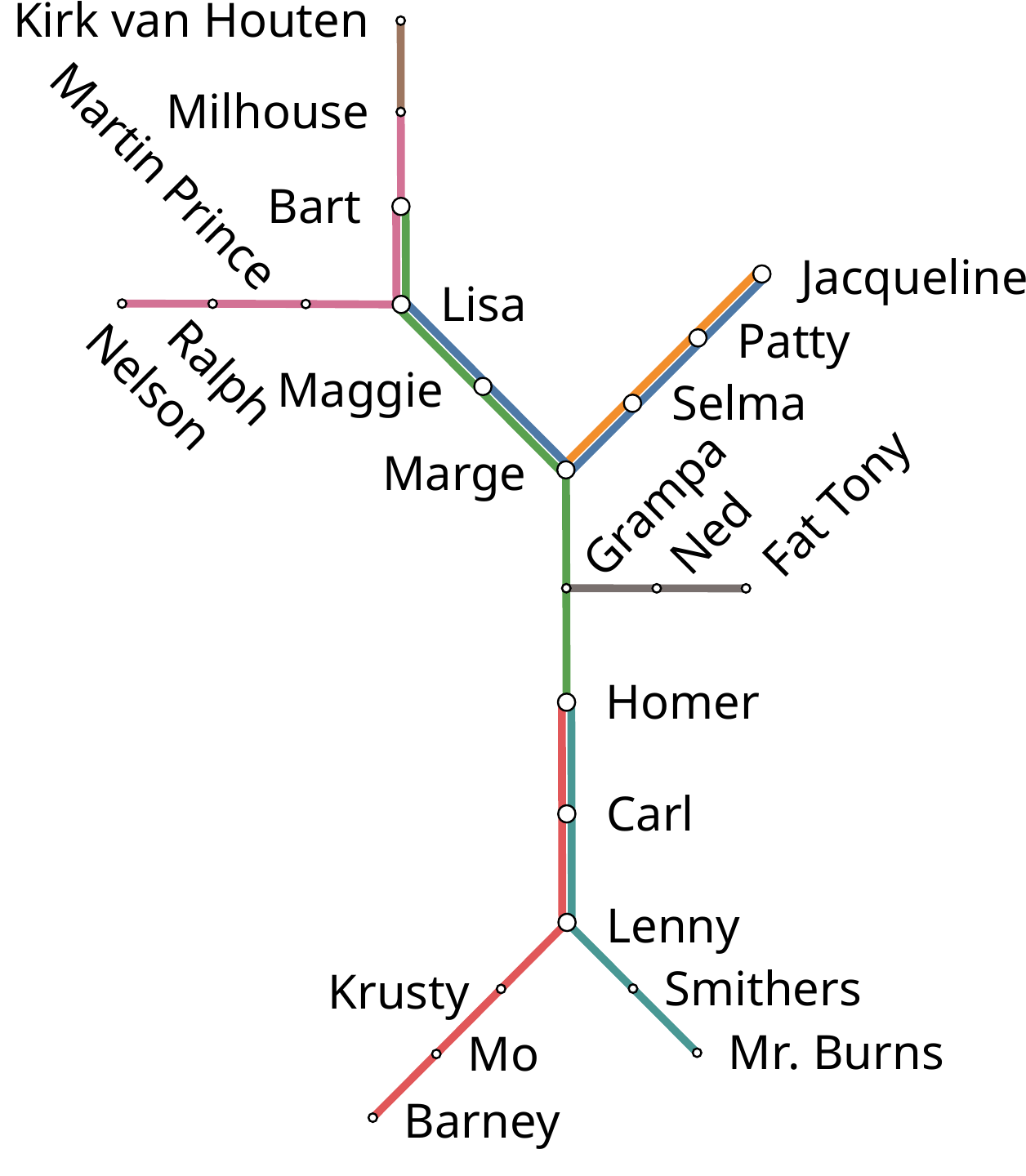}
        
        \smallskip
        
        \includegraphics[width=\linewidth,trim={32cm 0cm 4cm 20cm},clip]{img/simpsons-tree3.pdf}
        \subcaption{Hypergraph of 
        Simpson characters
        with a tree support.}
    \end{subfigure}
    \hfill
    \begin{subfigure}[b]{.62\textwidth}
        \centering
        \includegraphics[width=0.8\linewidth,trim={21cm 0 42.5cm 1cm},clip]{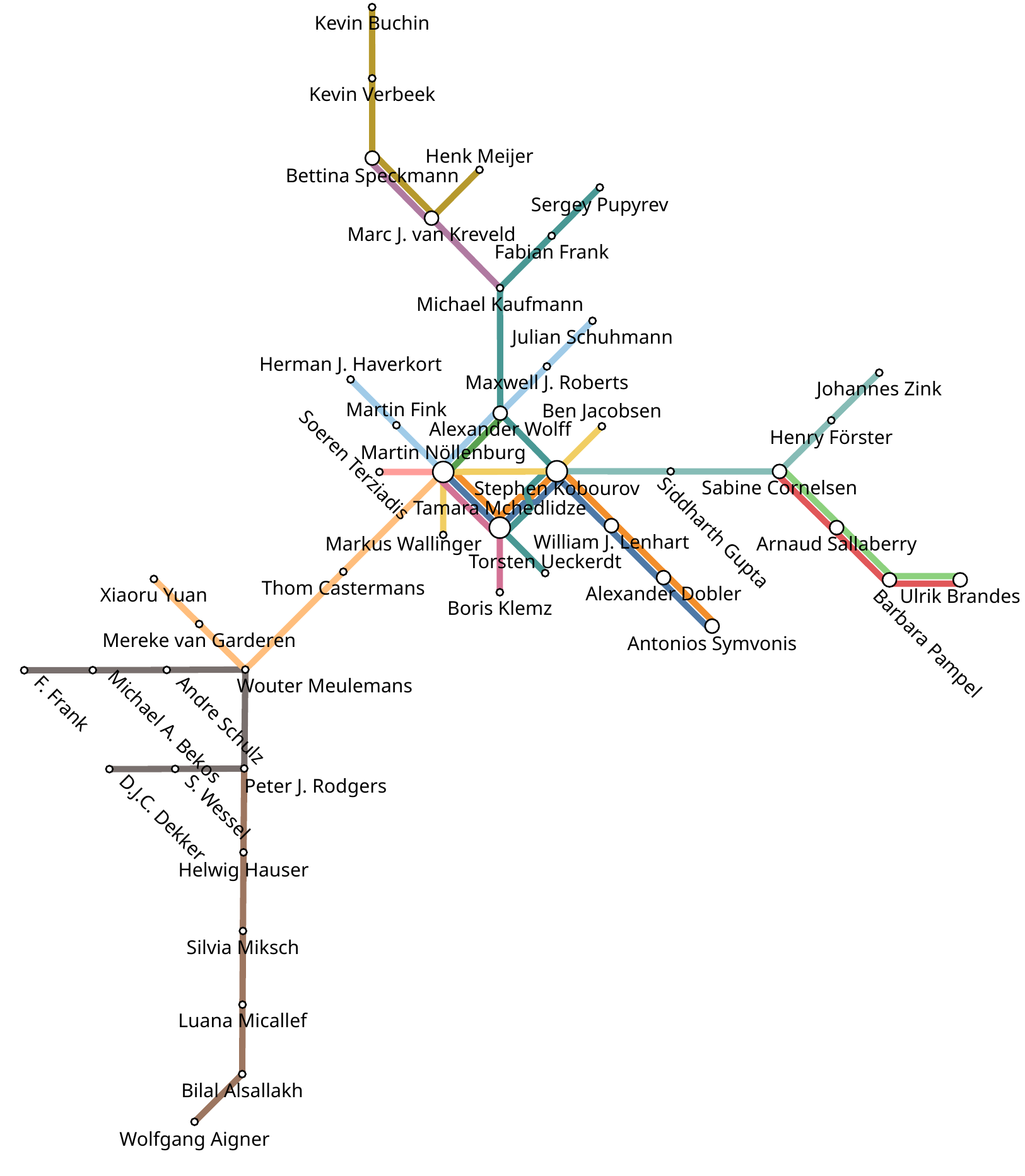}
        \subcaption{Largest component of the co-authorship network of this paper and its references.
        }
    \end{subfigure}
    \caption{%
    Metro map representations
    from~\href{http://metrosets.ac.tuwien.ac.at}{metrosets.ac.tuwien.ac.at}~\cite{metroMaps:ieee21}.}
\label{fig:metrosets}
\end{figure}

Planarity and bends are the most obvious measures of visual complexity for paths.
Let $n_i$ with $i = 0, 1, 2, \dots$ be the number 
of paths with $i$ bends.
We study how to draw a support graph subject to the following objectives (see also \cref{fig:caterpillar}):
\begin{enumerate}[label=(\roman*)]
    \item minimizing the \emph{total number of bends}, i.e.,
    $\sum_{i=0}^{\infty} i \cdot n_i$.
    \item minimizing the \emph{curve complexity}, i.e.,
    $\max_{i} i$ such that $n_i > 0$.
    \item lexicographically maximizing the \emph{bend vector}
    $(n_0, n_1, n_2, \dots)$,
    i.e., first maximizing the number $n_0$ of paths without bends,
    then maximizing the number $n_1$ of paths with one bend, then two bends, etc.
\end{enumerate}

We focus on straight-line drawings
where paths can only bend at vertices and \emph{orthogonal drawings} (see \cref{fig:4plane-example_intro})
where edges are drawn as polylines of horizontal and vertical line segments
and a path can bend at vertices and bends of edges.
The \emph{(planar) curve complexity}~of a path-based support is the minimum curve complexity over all its (planar) straight-line drawings. 

A special case are socalled \emph{linear hypergraphs/supports}   where any two hyperedges/paths share at most one vertex.
If in a planar embedding of a linear path-based support $(G, \p)$ no two paths touch, i.e., they are either disjoint or properly intersect (in a single vertex), the question if  $G$ is drawable such that no path has a bend corresponds to the $\exists\mathds R$-hard problem of pseudo-segment stretchability~\cite{schaefer/eppstein/gansner:gd09}. 
So, we
consider bend-minimization of path-based supports for restricted classes of graphs. Path-based tree supports (if they exist) have the advantage that (i) they can be computed in polynomial time~\cite{swaminathanWagner:siamJC94}, (ii) have the minimum number of edges among all supports, and (iii) display also intersections of any two hyperedges as a path. Cactus supports are a natural generalization of tree supports.

\begin{figure}[t]
    \centering
    \begin{minipage}[t]{0.46\linewidth}
        \centering
        \includegraphics[page=2]{img/caterpillar}
        \subcaption{Min. the total number of bends and lexicographically max. the bend vector.}\label{fig:cat2}
    \end{minipage}\hfil
    \begin{minipage}[t]{0.46\linewidth}
        \centering
        \includegraphics[page=2]{img/caterpillar2} \subcaption{Lexicographically max. the bend vector.}\label{fig:cat2:2}
    \end{minipage}
    
    \medskip
    
    \begin{minipage}[t]{0.46\linewidth}
        \centering
        \includegraphics[page=1]{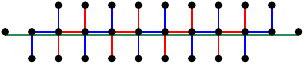}
        \subcaption{Min. the curve complexity.}\label{fig:catmax0}
    \end{minipage}
    \begin{minipage}[t]{0.46\linewidth}
        \centering
        \includegraphics[page=1]{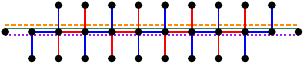}\hfil    \subcaption{ Min. the curve complexity and the total number of bends.}\label{fig:catmax0:2}
    \end{minipage}
    \caption{\label{fig:caterpillar}
        The caterpillars in (a) and (c) have the same set of paths.
        In (b) and (d) there are two more paths on the spine, affecting the optimality of the drawings.}
\end{figure}

\paragraph*{Related Work.}
Visualizing hypergraphs with the metro map metaphor has practical applications~\cite{metroMaps:ieee21} and gives rise to theoretical considerations~\cite{metroMapMetaphor:sofsem21}. Even though the choice of the paths for the hyperedges can impact the number of bends,
computing a visualization of a hypergraph in the metro map metaphor has often been considered as a combination of two formal problems 
of independent interest. 

The first problem is the computation of a suitable support graph. More precisely, given a hypergraph the goal is to decide if it admits an -- ideally path-based -- $\mathcal{G}$ support where $\mathcal{G}$ is a family of graphs. In this context, it is NP-complete to decide if there is a path-based planar support and NP-hard to find a path-based support with the minimum number of edges~\cite{Brandes:pathBased:da12}.
However, a path-based tree support~\cite{swaminathanWagner:siamJC94} and a (not necessarily path-based)\footnote{More generally, a graph is a \emph{support} of a hypergraph if every hyperedge induces a connected subgraph.} cactus support~\cite{blocks:iwoca10} can be constructed in polynomial time, if they exist.
Cactus supports can be used to store the set of all minimum cuts of an undirected graph \cite{dinitz:cactus76}. 
More work on computing planar supports can be found in~\cite{buchin_etal:jgaa11,planeSpatial:GMNY19,johnson/pollak:87,kaufmann:gd08,minimumTree:SWAT14}.

The second problem is to compute a high-quality geometric embedding for a given path-based support. The problem of embedding the support graph with few bends has been studied intensely from a practical point of view:
Bast, Brosi, and Storandt \cite{bastBrosiStorandt:2020} show how to  automatically generate metro maps with few bends on an octilinear grid using both, an exact ILP and a shortest-path based approximation algorithm. They also consider metro maps on triangular, octilinear, hexalinear, and ortho-radial grids~\cite{bastBrosiStorandt:sstd21}, while Nickel and N\"ollenburg~\cite{DBLP:conf/diagrams/NickelN20} generalized the mixed-integer programming approach of~\cite{noellenburg/wolff:2011} from octilinear grids to $k$ slopes through any vertex.  Hong et al.~\cite{hong_etal:2004} apply spring embedders to compute metro maps that mostly follow the octilinear grid. Fink et al.~\cite{fink_etql:GD12} guarantee bend-free metro-lines by using Bézier curves. 
Bekos et al.~\cite{bekos_etal:2022} study drawing supports for spatial hypergraphs such that the vertices are displaced on a rectangular or an ortho-radial grid. Among others they provide a simulated-annealing algorithm for tree-based supports and show that bend minimization is NP-hard.
From a theoretical point of view, Dobler et al.~\cite{dobler_etal:arxiv2024} consider the general problem of drawing hyperedges as lines or line-segments
where 
the order of the vertices of the hyperedges is not given. They show that it is $\exists\mathds R$-hard to decide whether a hypergraph admits such representations.
We remark that one of their proofs extends to the scenario where the order of vertices is fixed.

\paragraph*{Our Contribution.}

We investigate the second problem, i.e., minimizing the bends in a given path-based support graph -- in contrast to most earlier attempts, from a more theoretic point of view.

In \cref{sec:straightline}, we consider straight-line drawings of path-based  tree and cactus supports. For path-based tree supports one among our three objective functions can be optimized in polynomial time while the other two are hard:
Minimizing the curve complexity or lexicographically maximizing the bend vector is NP-hard (\cref{thm:treeGeneralCurveComplexityHard,thm:nphard-max0bendthen1bend}) even for very restricted classes of path-based tree supports. However, minimizing the total number of bends is polynomial-time solvable using  maximum weighted matching (\cref{thm:treeGeneral-total}). Deciding whether the curve complexity is 0 or 1 can also be done in polynomial time via a 2-SAT formulation (\cref{thm:treesGeneral-0-1}). 
We then study the parameterized complexity.
Deciding whether a path-based tree support has curve complexity $b$ is fixed-parameter tractable if parameterized by the number of paths (\cref{thm:treesGeneralFPT-paths}), the vertex cover number (\cref{thm:treesGeneralFPT-VC}), or the number of paths per vertex plus the curve complexity (\cref{thm:treesGeneralFPT-pathspervertex}).
The latter is a dynamic-programming approach that also yields FPT algorithms for lexicographically maximizing the bend vector (\cref{cor:maximizingNumberOf0BendAndSoOn}).

For a path-based cactus support we can test in near-linear time whether its planar curve complexity is zero (\cref{thm:cactusGeneral-0}).
Combining this with the dynamic program of \cref{thm:treesGeneralFPT-pathspervertex}, we get FPT algorithms for determining the curve complexity and for lexicographically maximizing the bend vector~of~cactus~supports.

Finally, in \cref{sec:orthogonal}, we show how to construct an orthogonal drawing of a path-based plane 4-graph support, minimizing the total number of bends on all paths. 
Proof details can be found in the appendix if marked with a~($\star$).

\section{Preliminaries}\label{sec:preliminaries}
We use standard graph terminology.
A \emph{binary tree} (of height~$h$) is a rooted tree
where every vertex has at most two children
(and there is no vertex with distance $h+1$ to the root).
A binary tree of height~$h$ is \emph{complete} if for $i \in \{0, \dots, h-1\}$,
there are $2^i$ vertices with distance~$i$ to the root
and \emph{perfect} if in addition it
has also $2^h$ leaves with distance~$h$ to the root.
A tree is a \emph{caterpillar} if removing its leaves makes it a path, its \emph{spine}.
All considered paths and cycles are \emph{simple}.

A problem~$\Pi$ is
\emph{fixed-parameter tractable} (FPT) or \emph{slice-wise polynomial} (XP)
parameterized by a parameter~$k$ if there is an algorithm
that solves an instance~$I$ of $\Pi$ in time
$\oh(f(k) \cdot |I|^c)$ or $\oh(|I|^{f(k)})$,
respectively, where $f$ is a computable function and $c$ is a~constant.

For a graph $G$, a vertex~$v$ and a subgraph~$C$ of~$G$,
we write $G-v$, $G+v$, and $G-C$ as a shorthand
for the graph arising from $G$ after removing~$v$,
adding~$v$, and removing all vertices of~$C$, respectively.
The size of a path-based support $(G,\mathcal P)$ is $\|\mathcal P\|:=\sum_{P \in \mathcal P}|P|$,
where $|P|$ is the number of edges on the path~$P$.
In a drawing of a path-based support,
we say that two incident edges are \emph{aligned} if they are contained in one line and disjoint except for their common endpoint.
An \emph{alignment (requirement)} of the set $E$ of edges of a support is a set of pairs of incident edges from $E$
such that for each vertex $v$ and for each edge $e$ incident to $v$ there is at most one other edge $e’$ incident to $v$ such that $e$ and $e’$ are a pair.
A \emph{realization} of an alignment is a drawing of the support in which the two edges of any pair are aligned.
\begin{figure}[t]
	\centering
	\begin{minipage}[b]{0.25\linewidth}
		\centering
		\includegraphics[page=5]{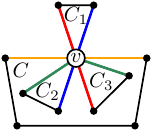}\hfil
		\subcaption{%
        \label{fig:cactusPlanarCC1}}
	\end{minipage}  \hfil
	\begin{minipage}[b]{0.25\linewidth}
		\centering
		\includegraphics[page=6]{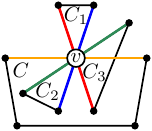}\\
		\subcaption{%
        \label{fig:cactusCC0}}
	\end{minipage}  \hfil
            \begin{minipage}[b]{0.45\textwidth}
        \centering
         \includegraphics[page=16]{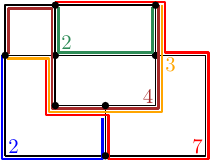}
    \subcaption{%
    \label{fig:4plane-example_intro}}
    \end{minipage}\hfil
	\caption{(a) Planar curve complexity 1 and (b) curve complexity 0 for the same cactus support. (c) orthogonal drawing of a plane 4-graph with 18 bends.}
	\label{fig:cactiPlanar1General0}
\end{figure}
Since we are not restricted to a specific grid, any alignment requirement for a path-based tree support can be realized in a planar way. Hence, the problem of drawing a path-based tree support is purely combinatorial as it suffices to compute a suitable alignment. 
In particular, the planar curve complexity of a tree equals its curve complexity. 
However, this is not true for path-based cactus supports (\cref{fig:cactiPlanar1General0}).
Moreover, each cycle needs at least three bends in any (non-degenerate) straight-line drawing.

\section{Straight-Line Drawings of Tree and Cactus Supports}\label{sec:straightline}

We first start with straight-line drawings of the support graph. 
The curve complexity in trees and cacti is in general unbounded: 
In a complete binary tree 
with paths from each leaf to the root, some path must bend at each inner vertex. Thus, in a complete binary tree with $n$ vertices there is a path with $\lfloor \log_2 n \rfloor -2 \in \Omega(\log n)$ bends; see \cref{fig:completeBinary}. In the cactus induced by the two paths $P: \langle v_0,v_1, v_2, \dots,v_{n-1} \rangle$ and $Q: \langle v_0,v_2,v_4,\dots v_{2\cdot \lfloor\frac{n-1}{2} \rfloor}\rangle$, the path $P$ bends at least $\lfloor \frac{n-1}{2} \rfloor \in \Omega(n)$ times; see \cref{fig:cactus}.

\begin{figure}[t]
	\centering
	\begin{minipage}[b]{0.6\linewidth}
		\includegraphics[page=3,scale=.95]{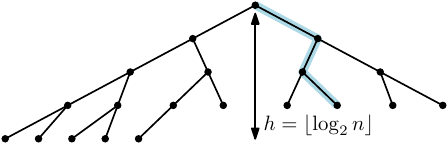}
		\subcaption{complete binary tree}
		\label{fig:completeBinary}
	\end{minipage}
    \hfill
	\begin{minipage}[b]{0.35\linewidth}
		\includegraphics[page=4]{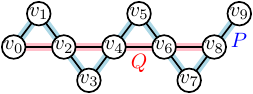}
		\subcaption{cactus}
		\label{fig:cactus}
	\end{minipage}
	\caption{The curve complexity of complete binary trees and cacti is unbounded.}
	\label{fig:lowerBound}
\end{figure}

\subsection{NP-Completeness}\label{sec:hardness}

We first show that minimizing the curve complexity is NP-hard. 

\begin{restatable}[\restateref{thm:treeGeneralCurveComplexityHard}]{theorem}{treeGeneralCurveComplexityHard}
	\label{thm:treeGeneralCurveComplexityHard}
	It is NP-complete to decide whether the curve complexity of a path-based tree support is at most $b$ even if 
	\begin{inparaenum}[(i)]
		\item $b=3$, every path has length at most~5,
        and the diameter of the support graph is at most~6; or
		\item the maximum degree of the support graph is at most~3.
	\end{inparaenum}
\end{restatable}

\begin{proof}[Overview \& Case (i).]
	Containment in NP is clear.
    To show NP-hardness, we reduce from 3-SAT.
	Given a 3-SAT formula $\Phi$ with $n$ variables and $m$ clauses, we construct a tree support 
	with (i) diameter~6 or (ii) maximum degree~3, respectively, that admits a drawing with at most $b=3$ or $b=2\cdot(\lceil\log_2 n\rceil + \lceil\log_2 m\rceil + 1)$ bends per path if and only if $\Phi$ has a satisfying truth assignment. See \cref{fig:3bendsHard} for Case~(i)  and \Cref{app:complete} for Case~(ii).
	
	For each variable $v_i$, $i=1,\dots,n$, the variable gadget is a star with central vertex~$x_i$ and two leaves labeled $v_i$ and $\neg v_i$.
	A clause gadget of a clause $c_j$, $j=1,\dots,m$, is a perfect binary tree of height two rooted at a vertex labeled~$c_j$ whose leaves are labeled with the literals of~$c_j$.
	One literal appears twice, once in each
	subtree rooted at the children of~$c_j$.
    Labels referring to a variable
    appear multiple times, once in a variable gadget and at most twice in each clause~gadget.

    \begin{figure}[t]
	\centering
	\includegraphics[page=1]{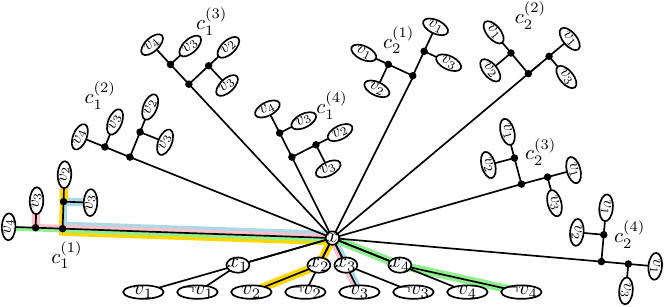}
	\caption{Graph with diameter~6 corresponding to clauses $c_1 = v_2 \lor v_3 \lor \neg v_4$
		and $c_2 = v_1 \lor \neg v_2 \lor v_3$,
		drawn corresponding to
        $v_1=v_3=v_4=\top$,  $v_2=\bot$.
		Among the $2\cdot 4 \cdot 4 = 32$ paths only the four to one copy of a clause gadget for $c_1$ are~drawn. }
	\label{fig:3bendsHard}
\end{figure}

	\bigskip
	\noindent\textit{Case (i).}
	For each clause~$c_j$, $j=1,\dots,m$, there are four copies of the clause gadget.
	A central vertex $x$ connects to all $4m$ central vertices of the clause gadgets and to all $n$ central vertices of the variable gadgets.
	The set of paths contains a path from any leaf labeled~$v_i$ or $\neg v_i$ of a clause gadget of a clause~$c_j$
	to the leaf labeled~$v_i$ or~$\neg v_i$ in the variable gadget of~$v_i$, depending on whether~$v_i$ or~$\neg v_i$ appears in~$c_j$.
	Clearly, the diameter of this graph is~6.
	
	In an optimal drawing,
    we may assume that, at every vertex of degree~3, exactly two edges are aligned.
    Hence, for each clause gadget $c$, 
    exactly one sub-path from $x$ to one leaf of~$c$ has two bends,
	while the remaining sub-paths ending in the other three leaves of~$c$ have zero or one bend each.
	There is at least one copy of each clause gadget such that each path ending there bends also in~$x$.
	Now consider the variable gadgets.
	We may assume that, for each $i=1,\dots,n$,
	the edge~$xx_i$ is aligned with either~$x_iv_i$ or~$x_i\neg v_i$ since $x_i$ is a degree-3 vertex.
	Interpreting the former as~$v_i$ being true,
	there is a satisfying truth assignment if and only if
	there is a drawing of curve complexity~3: 
	
	Assume first that there is a drawing with at most three bends per path.
	For a clause $c_j$, consider the clause gadget $H$ of $c_j$ such that each path ending in $H$ bends in $x$.
	Consider the leaf $u$ of $H$ such that the path $P$ ending in $u$ bends twice between $x$ and $u$.
	Then $P$ has already three bends and cannot have a further bend in the variable gadget.
	Thus, the respective literal is true.
	Assume now that $\Phi$ has a satisfying truth assignment.
	We align the variable gadgets according to the truth assignment.
	For a clause $c_j$ let $v_i$ be a variable that makes $c_j$ true.
	Draw the four clause gadgets associated with $c_j$
	such that the path from $x$ with two bends ends at a leaf labeled $v_i$.
	Then each path has at most 3~bends.

     \begin{figure}[t]
         \centering
         \includegraphics[page=17,scale=0.9]{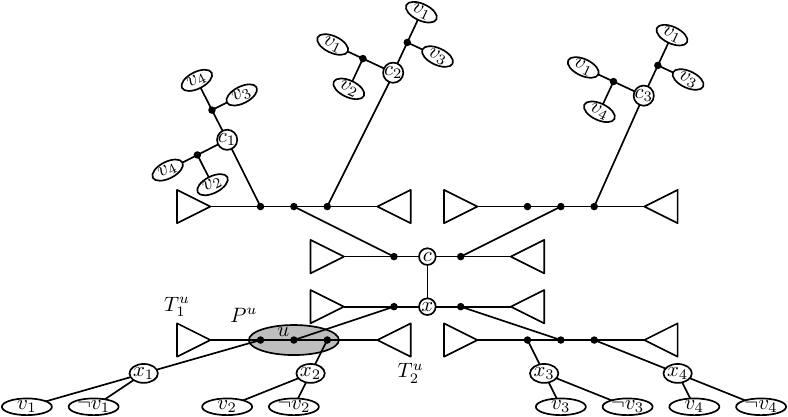}
         \caption{NP-hardness construction with maximum degree~3. 
             Triangles indicate perfect binary trees of height $\lceil \log_2 n \rceil + \lceil \log_2 m \rceil  + 1$ for $n$ variables and $m$ clauses; here $n=4$ and $m=3$. These perfect binary trees come in pairs $T^u_1$, $T^u_2$ with paths between all leaves in $T^u_1$ and all leaves in~$T^u_2$. 
         }
         \label{fig:treeGeneralCurveComplexityHardBoundedDeg}
     \end{figure}
     
     \bigskip
     \noindent\textit{Case (ii).}
     We consider this case in the appendix.
     The main idea is to realize the connections at the high-degree vertex~$x$ by binary trees.
     See \cref{fig:treeGeneralCurveComplexityHardBoundedDeg} for an illustration.
	 \qed
\end{proof}

Similarly, lexicographically maximizing the bend vector is also NP-hard.

\begin{restatable}[\restateref{thm:nphard-max0bendthen1bend}]{theorem}{nphardMaxZeroOneBends}
	\label{thm:nphard-max0bendthen1bend}
	It is NP-complete to decide if a path-based tree support admits a drawing in which $n_0$ paths have zero bends each and $n_1$ paths have one bend each even if the diameter of the support graph is at most four.
\end{restatable}

\subsection{%
Tree Supports~--~Exact Algorithms}\label{sec:trees}

The total number of bends can be minimized in polynomial time in trees.

\begin{restatable}[\restateref{thm:treeGeneral-total}]{theorem}{treeGeneralTotal}
    \label{thm:treeGeneral-total}
    Given a path-based tree support $(G=(V,E),\p)$,  
    an alignment of $E$ that yields a realization where the sum of bends in all paths is minimum, 
    can be computed in 
    $\oh(|V| \cdot |\p| + |V|^{2.5}\log^2(\max\{|V|,|\p|\}))$ time.
\end{restatable}
\begin{proof}
    For each vertex $v$, let  $K(v)$ be the complete graph on neighborhood $N(v)$ of~$v$.
    The weight of an edge $uw$ of $K(v)$ is the number of paths in $\p$ containing
    $uv$ and~$vw$.
    Let $M(v)$ be a maximum-weight matching on $K(v)$ and align $uv$ and $vw$ if and only if $uw \in M(v)$.
    We compute a maximum-weight matching~\cite{gabow:matching91} for each vertex.
    It is easy to see that such matchings yield an optimal solution for aligning edges at vertices.
    See \Cref{app:trees} for the running-time analysis. \qed
\end{proof}

The curve-complexity of a path-based tree support can be determined in polynomial time if it is very small.

\begin{theorem}\label{thm:treesGeneral-0-1}
    Given a path-based tree support $(G=(V,E),\p)$,
    it can be decided in 
    $\oh((\|\p\| + |V|^2) \cdot |V|)$ time
    whether its planar curve complexity is 0 or 1. 
\end{theorem}
\begin{proof}
    We use a 2-SAT formulation with a variable $x_{e_1e_2}$ for each pair $e_1,e_2$ of adjacent edges.
    We interpret $x_{e_1e_2}$ as true if and only if $e_1$ and $e_2$ are not aligned.
    For any three distinct edges $e_1,e_2,e_3$ sharing a common end vertex,
    we need the consistency condition that whenever $e_1$ is aligned to $e_2$ then $e_1$ cannot be aligned to $e_3$,
    i.e., the clause  $x_{e_1e_2} \vee x_{e_1e_3}$.
    By the same argument, we also have the clauses
    $x_{e_1e_2} \vee x_{e_2e_3}$ and $x_{e_1e_3} \vee x_{e_2e_3}$.
    If we insist on having no bends, then we need for any two adjacent edges $e_1,e_2$ on the same path the condition~$\neg x_{e_1e_2}$.
    If we aim for at most one bend per path, we add for any four edges $e_1,e_2,e_3,e_4$ on a path where $e_1$, $e_2$ are adjacent and $e_3$, $e_4$ are adjacent, the condition that at least one among the two pairs must be aligned, i.e., the clause  $\neg x_{e_1e_2} \vee \neg x_{e_3e_4}$.
    This yields a 2-SAT formulation with $\sum_{v \in V} \deg(v)(\deg(v)-1) \in \oh(|V|^2)$ variables and $\oh(\sum_{P \in \p}|P|^2+\sum_{v \in V} \deg^3(v)) \subseteq \oh(\|\p\|\cdot |V| + |V|^3)$ clauses.
    The number of constraints determines the running time
    since constructing the instance can be done in the same time
    and solving 2-SAT can be done in linear time. \qed
\end{proof}

The curve complexity is at most two if the support graph is a caterpillar:  draw the spine on a straight line as in
\cref{fig:catmax0:2}.
So, the curve complexity of a path-based caterpillar support can be determined in polynomial~time. 

\subsection{%
Tree Supports~--~Parameterized Algorithms}\label{sec:parameterizedTree}
We now consider the parameterized complexity of the problems.

\begin{theorem}\label{thm:treesGeneralFPT-paths}
    Determining the (planar) curve complexity for a path-based tree support $(G,\p)$ is in FPT parameterized by $|\p|$.
\end{theorem}
\begin{proof}
    For each leaf $v$ of~$G$, there is a path in $\p$ that ends in $v$, thus, the number of leaves is at most~$2|\p|$. Iteratively replace each vertex $v$ of degree~2 and its incident edges by a single edge connecting the neighbors of $v$, extending paths ending at $v$ to one of the neighbors. 
    Now the number of inner vertices is less than the number of leaves, i.e.\ the size of the resulting tree is less than $4|\p|$.
    This kernel can be solved brute-force testing all edge alignments. \qed
\end{proof}

\begin{theorem}
\label{thm:treesGeneralFPT-VC}
    Determining the (planar) curve complexity for a path-based tree support $(G,\p)$ is in FPT parameterized by the vertex cover number $k$.
\end{theorem}

\begin{proof}
    Let~$S$ be a minimum vertex cover of~$G$.
    If~$G$ has at most two vertices, then the curve complexity is~0. Otherwise, w.l.o.g.~$S$ contains no leaves of $G$.
    Let~$G'$ be the graph obtained from~$G$ by removing all its leaves.
    All leaves of $G'$ must be in $S$ and $S$ is a vertex cover of $G'$. In \Cref{app:trees}, we show inductively:
    
    \begin{restatable}[\restateref{cl:kernelSize}]{myclaim}{kernelSize}
    \label{cl:kernelSize}
    The number~$n'$ of vertices of $G'$ is at most 
    $2k-1$.
    \end{restatable}

    Thus, the number of distinct paths in the tree $G'$ is at most $\binom{2k-1}{2}$~-- one for each pair of vertices.     
    We show that it suffices to maintain at most six paths of $\p$ for each path of $G'$
    in order to obtain a support with the same curve complexity as $(G,\mathcal P)$.
    We call this set of retained paths $\tp$.
    Let~$\widetilde G$ be the graph induced by~$\tp$.
    In order to make sure that the curve complexity of $(G,\p)$ does not exceed the curve complexity of $(\widetilde G, \tp)$,
    we construct $\tp$ with the following property: 
    Let~$\widetilde\Gamma$ be a drawing of $(\widetilde G, \tp)$ with optimal curve complexity, let $\Gamma$ be a
    drawing of $(G, \p)$ with optimal curve complexity
 among all drawings extending $\widetilde\Gamma$,
    let $\pi' \in \p \setminus \tp$, and let $\pi$ be the subpath of $\pi'$ in $G'$. Then  $\tp$ contains a path $\widetilde \pi$ with the same subpath $\pi$ in $G'$ that has at least as many bends as $\pi'$. 
    We ensure~this using the following observations. 
    \begin{inparaenum}[(i)]
    \item
    Any path in $\p$ has at most two bends more than its subpath in $G'$.
    \item  
    Let $e=vx$ be an edge of~$G'$.
    If there are two leaves $u$ and $w$ adjacent to~$v$ in $G$ such that~$\p$ contains a path through~$e$ and $vu$
    and a path through~$e$ and~$vw$,
    then at least one of the two paths must bend at~$v$.
    \item
    If at least three paths share
    their sub-path in~$G'$ and extend on both ends to distinct leaves,
    at least one of these paths has two bends next to its leaves. 
    \end{inparaenum}
    
    We start with $\tp = \emptyset$. Let   $\pi$ be any path of $G'$ and let $v,v'$ be its end vertices.  Let $\p(\pi)$ be the set of paths of $\p$ whose intersection with $G'$ is precisely $\pi$.   Moreover, for $i=0,1,2$, let $\p_i(\pi)$ be the set of paths in $\p(\pi)$ containing $i$ leaves of $G$. 
    Let $L_v(\pi), L_{v’}(\pi)$ be the set of leaves adjacent to $v$ and $v’$, respectively, such that some path of $\p_2(\pi)$ ends there and assume w.l.o.g.\ that $|L_v(\pi)| \leq |L_{v’}(\pi)|$. We distinguish several cases. 
    \begin{inparaenum}[(a)]
    \item
    If $\p_1(\pi) \cup \p_2(\pi) = \emptyset$ and $\p_0(\pi) \neq \emptyset$ add $\pi$ to $\tp$. 
        \item
    If $\p_2(\pi)$ contains at most one path $\pi$, we add $\pi$ and up to four paths from $\p_1(\pi)$ to $\tp$, up to two containing a leaf incident to $v$ and up to two containing a leaf incident to $v'$.
        \item If $|L_v(\pi)|\leq 3$ then add for each vertex $u \in L_v(\pi)$ up to two paths
    of  $\p_2(\pi)$ ending in $u$ to~$\tp$.
    \item If there are two vertices $u_1,u_2 \in L_v(\pi)$ such that at least two paths of $\p_2(\pi)$ end in $u_1$ and $u_2$ each, add two for both of them to $\tp$; proceed symmetrically for $L_{v'}(\pi)$.
    \item\label{case:triple} Otherwise, if there are six vertices $u_1,u_2,u_3 \in L_v(\pi)$ and $u'_1,u'_2,u'_3 \in L_{v'}(\pi)$ such that there are $u_i$-$u'_i$-paths for $i=1,2,3$ in $\p_2(\pi)$ add these three paths to $\tp$.
    \item In the remaining cases, $|L_v(\pi)|\geq 4$, $|L_{v'}(\pi)|\geq 4$, and there is at most one vertex $u \in L_v(\pi)$ and at most one vertex $u' \in L_{v'}(\pi)$ where more than one path ends.
    (If in all vertices of $L_v(\pi)$ ($L_{v'}(\pi)$) only one path ends, let $u$ ($u'$) be any of them.)
    In this case, we add up to two paths of  $\p_2(\pi)$ ending in $u$ or $u'$, respectively, to $\tp$, and, if it exists, additionally one path of $\p_2(\pi)$ with no end vertex in $\{u,u'\}$. Observe that there cannot be two such paths since otherwise Case~\ref{case:triple} would hold.
    \end{inparaenum}
    
    Finally, remove all leaves of~$G$ where none of
    the paths of $\tp$ ends. Let the resulting graph be $\widetilde G$.
    This yields a kernel $(\widetilde G,\tp)$ with $\oh(k^2)$ vertices and $\oh(k^2)$~paths having the same curve complexity as $(G,\p)$. \qed
    \end{proof}

Path $p$ \emph{passes through}  vertex $v$, if $v$ is contained in $p$, but not as an end~vertex.

\begin{theorem}
\label{thm:treesGeneralFPT-pathspervertex}
    Determining the (planar) curve complexity 
    for a path-based tree support $(G,\p)$ is in
    XP parameterized by 
    the maximum number $k$ of paths in $\p$ 
    passing through a vertex.
    Deciding if the (planar) curve complexity is
    $b$ is in FPT parameterized by~$k+b$.
\end{theorem}

\begin{proof}
	We use bottom-up dynamic programming. Root $G$ at any vertex. For each vertex~$v$, we consider feasible drawings of the subtree $G_v$ rooted at $v$ including the edge $e_v$ between~$v$ and its parent $p_v$. For each such drawing, we store a record that contains for each path $\pi$ of $\p(v) := \{\pi \in \p \mid e_v \in \pi\}$ the number~$b_{\pi}$ of bends in~$G_v$.
    Since $b_{\pi} \in \{0, 1, \dots, b\}$,
    this yields $\oh(b^k)$ different records ($b$ is upper bounded by~$n$).
    Let $B(v)$ be the set of records that we store for $v$. For each record $\langle b_{\pi} \mid \pi \in \p(v)\rangle$ in $B(v)$, we also store the minimum curve complexity $C_v(\langle b_{\pi} \mid \pi \in \p(v)\rangle)$ that can be achieved in $G_v$ with $b_{\pi}$ bends on $\pi \in \p(v)$.
	
	If $v$ is a leaf, $B(v)$ contains a single record where all paths have zero bends.
    Let $v$ be  a non-leaf vertex. 
	Let $c_i$, $i=1,\dots,s$ denote the \emph{active} children~of~$v$, i.e., the children of $v$ that occur on a path in $\p$ with internal vertex $v$.
	Since there are at most $k$ paths through $v$ and each such path contains at most~two edges incident to $v$, it follows that $s \leq 2k$. We have already recursively computed the set $B(c_i)$ of all records for $c_i$. We now try each combination of records in $B(c_1),\dots,B(c_s)$ combined with each possible alignment of edges in  $\{vc_1,\dots,vc_s,vp_v\}$. More precisely, for $i=1,\dots,s$ let $\langle b_{\pi}^i \mid \pi \in \p(c_i)\rangle$ be a record in $B(c_i)$ and consider a possible alignment of the edges in  $\{vc_1,\dots,vc_s,vp_v\}$. We compute the record in $B(v)$ and the respective curve complexity $C_v$ as follows. Set $b_{\pi} = 0$ if $\pi$ starts at $v$ and contains $e_v$. For each path~$\pi$ through $v$ do the following:
	\begin{enumerate}
        \item If $\pi$ contains $c_i,v,p_v$, we set $b_{\pi} = b_{\pi}^i$ if  we aligned the edges $vc_i$ and $vp_v$, and $b_{\pi} = b_{\pi}^i+1$ otherwise. If $b_{\pi} > b$, we reject the resulting record.
		\item\label{it:pathCompleted} If $\pi$ contains $c_i,v,c_j$, for two children $c_i$ and $c_j$ of $v$, we set $b_{\pi} = b_{\pi}^i + b_{\pi}^j$ if we aligned the edges $vc_i$ and $vc_j$, and $b_{\pi} = b_{\pi}^i + b_{\pi}^j + 1$ otherwise. If $b_{\pi} > b$, we reject the resulting record.
	\end{enumerate}	
		If we did not reject the record when handling a path $\pi$ through $v$, add the record $\langle b_{\pi} \mid \pi \in \p(v)\rangle$ to $B(v)$.  
		Obtain the curve complexity from the curve complexities of the newly combined paths, the respective records from $c_1,\ldots,c_s$, and the maximum curve complexity $C_{\max}$ used for inactive child trees.  Namely, we set
		\begin{equation*}
			C_v(\langle b_{\pi} \mid \pi \in \p(v)\rangle)=\max(
				\max_{i=1,\dots,s}C_{c_i}(\langle b_{\pi}^i \mid \pi \in \p(c_i)\rangle),
				\max_{\pi \textup{ path through }v} b_{\pi},
				C_{\max}).
		\end{equation*}
	
	In the process, we may create a record with the same number of bends per path through $e_v$ several times.
	We update the curve complexity if the newly computed record achieves  better curve complexity, otherwise, we keep the old solution. Amongst the valid solutions for the root, we select the one optimizing the curve complexity. The correctness follows from the fact that we exhaustively considered all possibilities, i.e.,  it is obvious that each subtree of a vertex $v$ can be aligned with at most one sibling subtree or with the parent $p_v$ of $v$. Otherwise, new bends are created as computed in the algorithm. Finally, a drawing always exists for each record that does not exceed the bend limit $b$ along a single path.
	
	The runtime is dominated by the number of alignments of subtrees and by the number of records.
    Each of the at most~$k$ subtrees may be aligned
    to another subtree or do nothing.
    This equals the number of matchings in the complete graph~$K_k$,
    for which a (non-tight) upper bound is $k! \in 2^{\oh(k \log k)}$.
	The number of records is in $\oh((b^k)^{2k})$ in each of the $n$ steps as seen before.
    In particular, note that the optimal solution for each inactive child is always the same since we do not need to align it anymore.
    The computations for each combination can be performed in polynomial time. Hence, we obtain an XP algorithm if $b$ is unbounded (i.e., $b \in \oh(n)$) and an FPT algorithm if $b$ is a parameter.
    We want to point out that the proof is constructive and that an alignement yielding the computed curve complexity can be stored throughout the construction. \qed
\end{proof}

Keeping track of the number of paths with a certain number of bends yields:

 \begin{corollary}
    \label{cor:maximizingNumberOf0BendAndSoOn}
    Given a path-based tree support $(G, \p)$,
    finding a drawing that lexicographically maximizes the bend vector
    is in XP parameterized by $k$ where $k$ is the maximum number of paths in $\p$ 
    passing through a vertex.
    Moreover, there is an FPT algorithm parameterized by $k+b$
    where $b$ denotes the maximum number of bends allowed on each path.
\end{corollary}
\begin{proof}
    We slightly modify the dynamic program from \cref{thm:treesGeneralFPT-pathspervertex}.
    Instead of evaluating the quality of the record $\langle b_p \mid p \in \p(v)\rangle$ on the curve complexity $\mathcal{C}_v(\langle b_p \mid p \in \p(v)\rangle)$, we evaluate it in terms of the number of paths with exactly $\beta$ bends that are entirely located within the already finished subtree. Thus, the best record first maximizes the number $n_0$ of 0-bend paths, then amongst records  with $n_0$ $0$-bend paths, it maximizes the number $n_1$ of $1$-bend paths, and so on, up to $b$ (in the FPT case) or $n$ (in the XP case).
    Observe that we still keep a record for each combination of $\langle b_p \mid p \in \p(v)\rangle$ for each allowed number~$\beta$ of bends, i.e., we still explore the full solution space. This shows the XP algorithm in terms of $k$ and the FPT algorithm in terms of $k+b$.
    \qed
\end{proof}

\subsection{%
Cactus Supports}\label{sec:cactus}

We show in \Cref{sec:linear} that every linear path-based cactus support has curve-complexity 0 and how to determine in linear time whether its planar curve-complexity is 0. We use this result for the following theorem.

\begin{restatable}[\restateref{thm:cactusGeneral-0}]{theorem}{cactusGeneral}
	\label{thm:cactusGeneral-0}
	Given a path-based cactus support $(G=(V,E),\p)$,
	it can be decided in 
	near-linear 
	time
	whether its (planar) curve complexity is~0.
\end{restatable}
\begin{proof}
	If there are three edges $e$, $e_1$, and $e_2$ incident to a vertex~$v$ and two paths in $\p$, one containing $e$ and $e_1$ and the other containing $e$ and $e_2$ then there is no drawing without bends. Otherwise we merge paths that share an edge into one path. Thus, in the obtained cactus support $(G=(V,E),\p')$ no two paths share an edge and it has (planar) curve complexity zero if and only if  $(G=(V,E),\p)$ does. If there are two paths in $\p'$ that intersect twice, then the two paths would contain a cycle and, thus, cannot both be drawn on a straight line. Otherwise, the support is linear and we can solve the problem in linear time.
	A na\"ive implementation of this algorithm yields cubic running time.
	Using union-find data structures, we prove the near-linear running time in the appendix. \qed
\end{proof}

Using \cref{thm:cactusGeneral-0} and also keeping track of the number of bends per cycle, we generalize the dynamic program from trees to cactus~supports.

\begin{restatable}[\restateref{thm:extensionToCacti}]{theorem}{extensionToCacti}
    \label{thm:extensionToCacti}
    \cref{thm:treesGeneralFPT-pathspervertex,cor:maximizingNumberOf0BendAndSoOn} generalize to path-based cactus~supports.
\end{restatable}

\section{Orthogonal Drawings for Plane 4-Graph Supports}\label{sec:orthogonal}

Metro maps are often layed out on an octilinear grid. Here, we consider the more restrictive orthogonal grid. See \cref{fig:4plane-example_intro}.

\begin{restatable}[\restateref{thm:tamassia}]{theorem}{tamassia}
\label{thm:tamassia}
    Given a path-based plane 4-graph support $(G,\p)$, an orthogonal drawing of $G$ with the property that 
    the sum of the bends in all paths is minimimum among all orthogonal drawings of $G$ with the given embedding can be computed in near-linear time.
\end{restatable}

\begin{proof}[Sketch]
    We adapt methods for bend minimization for orthogonal drawings of plane $4$-graphs. The key observation is that we now count bends on paths of $\p$ instead of bends on edges of the support graph $G$. This leads to the following two key differences in the evaluation of an orthogonal drawing of $G$:
        First, bends on an edge $e$ actually correspond to $k(e)$ bends where $k(e)$ denotes the number of paths in $\p$ containing $e$.
        Second, while the alignment of edges at a vertex $v$ is not directly contributing to the bends in the drawing of the support graph, vertex $v$ might be an interior vertex of some paths in $\p$ and we have to account for the number of paths that bend at each vertex $v$. To this end, we observe that the evaluation depends on the degree of $v$: 
        \begin{inparaenum}
            \item If $v$ has degree $1$, it is not an interior vertex of any path. 
            \item If $v$ has degree $4$, there is no choice to be made.
            \item If $v$ has degree $2$, we can choose whether to align its two incident edges (in this case, no path bends at $v$) or not to align them (in this case, all paths with interior vertex $v$ bend at $v$).
            \item If $v$ has degree $3$, we can align exactly one pair of its three incident edges. All paths traversing the other two pairs of incident edges bend at $v$. 
        \end{inparaenum} 
    In \Cref{app:tamassia}, we describe how to adjust the approach of Tamassia \cite{tamassia:87} for bend-minimum orthogonal drawings to correctly evaluate the number of bends along the paths of $\p$ which results in a near-linear time~\cite{brand_etal:focs23} algorithm for minimizing the total number of bends along the paths in $\p$.\qed
\end{proof}

\section{%
    Open Problems}
    \label{sec:op}
Our results motivate open problems for drawing path-based tree or cactus supports with~few bends. What is the complexity of
\begin{inparaenum}[(a)]
    \item minimizing the total number of bends for a path-based cactus support,
    \item minimizing the curve complexity
    and lexicographically maximizing the bend vector
    for a plane 4-graph support,
    \item deciding whether the curve complexity is two for a tree support, or one or two for a cactus support,
    \item maximizing the number of paths with no bends,
    \item lexicographically maximizing the bend vector for trees of maximum degree 3,
     and
    \item computing a path-based cactus support.
    Moreover, \item is the curve complexity W[1]-hard if parameterized by the maximum number of paths through~a~vertex?
\end{inparaenum}

\begin{credits}
	\subsubsection{\ackname} 
	This work was initiated at the Annual Workshop on Graph and Network Visualization (GNV~2024), Heiligkreuztal, Germany, June 2024.
	
	\subsubsection{\discintname}
	 The authors have no competing interests to declare that are
	relevant to the content of this article. 
\end{credits}
\bibliographystyle{splncs04}
\bibliography{references}

\clearpage

\appendix

\section{Appendix: Missing Proof Details for NP-Completeness}\label{app:complete}

\treeGeneralCurveComplexityHard*
\label{thm:treeGeneralCurveComplexityHard*}
\begin{proof}[of Case (ii).]
	We modify the construction to have maximum degree~3; see \cref{fig:treeGeneralCurveComplexityHardBoundedDeg} for an illustration.
	There is one copy of each clause gadget. 
	We again have a central vertex~$x$ but, instead of directly connecting it with all variable gadgets,
	we connect it via a binary tree.
	More precisely, $x$ is the root of a complete binary tree having the central vertices of the variable gadgets as leaves
	(the binary tree is perfect if $n$ is a power of~2).
	Similarly, we have a second central vertex~$c$ that connects directly to~$x$ and connects the clause gadgets via a binary tree.
	More precisely, $c$ is the root of a complete binary tree having the central vertices of the clause gadgets as leaves
	(the binary tree is perfect if $m$ is a power of~2).
	Again, the set of paths contains a path from any leaf labeled~$v_i$ or $\neg v_i$ of a clause gadget of a clause~$c_j$
	to the leaf labeled~$v_i$ or~$\neg v_i$ in the variable gadget of~$v_i$, depending on whether~$v_i$ or~$\neg v_i$ appears in~$c_j$. 
	All variable--clause paths have the same length, namely, the number of inner vertices of an $x_i$--$c_j$-path is $\lceil\log_2 n\rceil + \lceil\log_2 m\rceil$.
	
	To ensure that all variable--clause paths
	bend the same number of times when traversing the two binary trees, we replace the internal vertices by \emph{alignment gadgets} which ensure that any variable--clause path will bend at any internal vertex except possibly within the variable or the clause gadget.
	The alignment gadget for vertex $u$ consists of two perfect binary trees $T^u_1$ and $T^u_2$ of height $\lceil\log_2 n\rceil + \lceil\log_2 m\rceil + 1$ each, connected by a path $P^u$ of length four. 
	The set of paths contains a path from every leaf in $T^u_1$ to every leaf in $T^u_2$.
	One of these paths must bend $b$ times within $T_1$ and $T_2$, so  the edges of $P^u$ must be aligned
	at the inner vertices of $P^u$ in any drawing with at most $b$ bends per path.
	Thus, each variable--clause path bends $2\cdot(\lceil\log_2 n\rceil + \lceil\log_2 m\rceil)$ times
	strictly between a vertex labeled~$x_i$ and a vertex labeled~$c_j$. 
	
	Analogously to Case~(i), within each clause gadget there is one path $P$ that bends twice, say at the path to a vertex labeled $v_i$. So there is a drawing with at most $b$ bends per path if and only if $P$ does not bend at the root $x_i$ of the respective variable gadget.
	This shows the equivalence of satisfying truth assignments
	and drawings with at most $b$ bends. \qed
    \end{proof}

\nphardMaxZeroOneBends*
\label{thm:nphard-max0bendthen1bend*}
\begin{proof}
    Since the problem on trees is purely combinatorial, containment in NP is clear.
    We prove NP-hardness by a reduction from Not-All-Equal-3-SAT (NAE-3-SAT)
    where the goal is to determine an assignment of truth values to the variables
    so that each clause $(\ell_1 \lor \ell_2 \lor \ell_3)$ contains one or two satisfied literals from the literals $\{\ell_1, \ell_2, \ell_3\}$, i.e.,
    in a valid solution for NAE-3-SAT, both $(\ell_1 \lor \ell_2 \lor \ell_3)$ and $(\neg \ell_1 \lor \neg \ell_2 \lor \neg \ell_3)$ are true.
    Given a NAE-3-SAT formula $\Phi$ with $n$ variables and $m$ clauses, we construct a tree support $(G=(V,E),\p)$
    and determine values $n_0$ and $n_1$ as follows; see also \cref{fig:min0and1bendhard} for an illustration.
    
    \begin{figure}
        \centering
        \includegraphics[page=21]{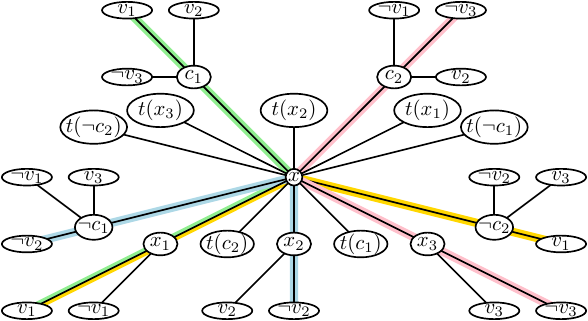}
        \caption{Graph corresponding to NAE-3-SAT instance with  clauses $c_1 = (v_1 \lor v_2 \lor \neg v_3)$
            and $c_2 = (\neg v_1 \lor v_2 \lor \neg v_3)$.
            The drawing corresponds to the truth assignment where $v_1$, $v_2$, and $v_3$ are true.
            The four paths with 1 bend are highlighted.}
        \label{fig:min0and1bendhard}
    \end{figure}
    
    As in the proof of \cref{thm:treeGeneralCurveComplexityHard},
    there is a variable gadget for each variable $v_i$, $i=1,\dots,n$,
    which is a star consisting of a central vertex~$x_i$ together with two leaves labeled $v_i$ and $\neg v_i$.
    Moreover, for each clause $c_j=(\ell_1\lor \ell_2\lor \ell_3)$ with $j=1,\ldots,m$, we create two clause gadgets, namely,
    the \emph{positive clause gadget} consisting of a star on $4$ vertices $\ell_1,\ell_2,\ell_3,c_j$ with center $c_j$ and
    the \emph{negative clause gadget} consisting of a star on $4$ vertices $\neg \ell_1,\neg \ell_2,\neg \ell_3,\neg c_j$ with center $\neg c_j$.
    In addition, we introduce a central vertex $x$ which is connected to each central vertex $x_i$ with $i=1,\ldots,n$ of the $n$ variable gadgets
    and to all central vertices~$c_j$ and~$\neg c_j$ with $j=1,\ldots,m$ of the $2m$ clause gadgets.
    Finally, we add $n+2m$ \emph{tail vertices} $t(x_i)$ with $i=1,\ldots,n$ and $t(c_j)$ and $t(\neg c_j)$ with $j=1,\ldots,m$ that are also connected to $x$.
    
    This concludes the construction of $G$.
    Next, we discuss the set of paths $\p$.
    Namely, $\p$ consists of two sets of paths $\p_0$ and $\p_1$.
    For each variable gadget, we have a path $x_i,x,t(x_i)$ in $\p_0$ for $i=1,\ldots,n$, and,
    for each clause gadget, we have two paths $c_j,x,t(c_j)$ and $\neg c_j,x,t(\neg c_j)$ in $\p_0$ for $j=1,\ldots,m$.
    Finally, for $\p_1$ we add a path between each literal $\ell_i$ occurring in a clause gadget and the same literal (i.e., $v_i$ or $\neg v_i$) in the variable gadget of~$v_i$.
    Also, we add a path between each literal $\neg \ell_i$ occurring in a clause gadget with $\neg \ell_i$ in the variable gadget of~$v_i$.
    This concludes the construction of $\p$.
    Observe that $|\p_0|=n+2m$ whereas $|\p_1|=6m$.
    We finish the reduction by setting $n_0=n+2m$ and $n_1=2m$.
    
    It remains to show that $(G,\p)$ admits a drawing with $n_0$ paths each having $0$ bends and $n_1$ paths each having $1$ bend if and only if $\Phi$ is a yes-instance.
    First, assume that $\Phi$ is a yes-instance.
    We align the paths $x_i,x,t(x_i)$ for $i=1,\ldots,n$ and the paths $c_j,x,t(c_j)$ and $\neg c_j,x,t(\neg c_j)$ for $j=1,\ldots,m$.
    This creates $n_0=n+2m$ paths with $0$ bends as required.
    Then, if $v_i$ is true in the satisfying truth assignment, we align edges $v_ix_i$ and $x_ix$, otherwise,
    i.e., if $v_i$ is false, we align edges $\neg v_ix_i$ and $x_ix$.
    Moreover, for clause $c_j=(\ell_1 \lor \ell_2\lor \ell_3)$ let w.l.o.g.\ $\ell_1$ be true and $\ell_2$ be false.
    Then, we align the edges $c_j\ell_1$ and $xc_j$ as well as $\neg c_j\neg \ell_2$ and $x\neg c_j$.
    As a result, the paths ending at $\ell_1$ and $\neg \ell_2$ have one bend each (namely at~$x$) resulting in $n_1=2m$ paths with one bend.
    The remaining paths (which all pass through a clause gadget) have one bend in $x$ and one bend in the central vertex of the clause gadget.
    
    Second, assume that $(G,\p)$ admits a drawing with $n_0$ paths each having $0$ bends and $n_1$ paths each having $1$ bend.
    Observe that at most $\max \{2m, n\}$ paths of $\p_1$ can be realized with $0$ bends
    since in each clause gadget, three paths of $\p_1$ share the edge between $x$ and the central vertex of the clause gadget
    but then continue along a different edge towards a vertex representing a literal.
    Similarly, each of the $n$ variable gadgets can be aligned with only one clause gadget, in which only one path can have $0$ bends.
    Assume for a contradiction that $k > 0$ paths from $\p_1$ are drawn with $0$ bends.
    Each of these $k$ paths traverses a different clause gadget with center $c^*$ so that the two edges $c^*x$ and $xx_i$ for some variable $v_i$ are aligned.
    Thus, the edges $c^*x$ and $t(c^*)x$ and the edges $x_ix$ and $t(x_i)x$ cannot be aligned which implies that
    $2k$ paths from $\p_0$ are drawn with more than $0$ bends.
    So, there are at most $|P_0|-(2k)+k=n+2m-k<n_0$ paths with $0$ bends in $\p$; a contradiction.
    Thus, all paths in $\p_0$ have $0$ bends.
    
    Hence, it remains to discuss the paths with $1$ bend which are $2m$ paths from $\p_1$.
    Observe that each path in $\p_1$ bends at $x$ since it contains an edge $x_ix$ and an edge $xc_i$ or $x\neg c_i$ while $x_ix$ is aligned with $t(x_i)x$.
    Moreover, since, in each clause gadget, three paths of $\p_1$ still share the edge between $x$ and the central vertex of the clause gadget,
    but then continue along a different edge towards a vertex representing a literal,
    we have that there is exactly one path of $\p_1$ with one bend traversing each clause gadget. 
    
    We now assign the variable $v_i$ to true if and only if edges $v_ix_i$ and $x_ix$ are aligned.
    Consider now a clause $c_j$.
    Each of its two clause gadgets must be traversed by a path bending only at $x$.
    W.l.o.g., let $P=v_i,x_i,x,c_j,v_i$ and $P'=\neg v_k,x_k,x,\neg c_j,\neg v_k$ be these paths for some variables $v_i$ and $v_k$,
    that is, $v_i$ and $v_k$ are literals in $c_j$ (the cases where one or both of the occurring literals are negative are symmetric).
    Since $P$ and $P'$ bend at $x$, we have that $v_ix_i$ and $x_ix$ are aligned, i.e., we have set $x_i$ to true.
    Moreover, we have that $\neg v_kx_k$ and $x_kx$ are aligned, i.e., we have set $x_k$ to false.
    Thus, $c_j$ contains both a true and a false literal, i.e., $c_j$ is satisfied.
    This concludes the proof. \qed
\end{proof}

\section{Appendix: Missing Proof Details for Trees}\label{app:trees}

\treeGeneralTotal*
\label{thm:treeGeneral-total*}

\begin{proof}
    It remains to analyze the running time.
    To construct the complete graphs $K(v)$ for each vertex $v \in V$ and to compute their edge weights,
    we use
    $\oh(\|\p\| + |V|^2)$  time in total:
    First, construct the complete graphs $K(v)$ for each $v \in V$ in overall $\oh(|V|^2)$ time.
    Initialize all edge weights with~0.
    Note that each pair of complete graphs constructed in this process is edge-disjoint.
    We assume that the edges of the input graph are enumerated and stored in an $|E|\times |E|$ array storing pointers to the edges in the matching graph for a direct access.
    Since, $|E| \in \oh(|V|)$, this table has size $\oh(|V|^2)$ and can be computed in time $\oh(|V|^2)$.
    Then, we go in $\oh(\|\p\|)$ time over all paths
    and add for each pair of adjacent edges $uv$ and $vw$ on a path
    a weight of~1 to the edge $uw$ in the corresponding complete graph $K(v)$.

    In a graph with $n$ vertices and $m$ edges, which have positive integral edge weights of maximum weight~$W$,
    a maximum-weight matching can be computed in
    $\oh(m \sqrt{n \log n \alpha(m, n)} \log(nW))$ time~\cite{gabow:matching91}, where $\alpha$ is the inverse Ackermann function.
    In our case $W \in \oh(|\p|)$, $n \in \oh(|N(v)|)$, $m \in \oh(|N(v)|^2)$,
    which leads to a running time in $\oh(|N(v)|^{2.5} \log(|N(v)| |\p|) \sqrt{\log N(v) \alpha(|N(v)|^2, |N(v)|)})$ for $v$'s maximum-weight matching.
    For all vertices, this can be upper bounded by
    $\oh(|V|^{2.5}\log(\max\{|V|,|\p|\}) \sqrt{\log |V| \alpha(|V|^2, |V|)})$,
    which we in turn upper bound by
    $\oh(|V|^{2.5}\log^2(\max\{|V|,|\p|\}))$.
    
    Since $\max\{|V|,|\p|\} \leq \|\mathcal P\| \leq |V| \cdot |\p|$,
    this gives a running time in 
    $\oh(\|\p\| + |V|^2 + |V|^{2.5}\log^2\|\p\|) \subseteq \oh(|V| \cdot |\p| + |V|^{2.5}\log^2(\max\{|V|,|\p|\}))$ for both steps combined. \qed
\end{proof}

\kernelSize*
    \label{cl:kernelSize*}

    \begin{proof}
    We use  induction on $n'$.
    First, observe that if $n' = 1$, then $k = 1$ and $2k - 1 = 1 = n'$.
    Now, assume that $n' > 1$ and let~$v$ be a leaf of~$G'$ and let~$w$ be the neighbor of~$v$.
    We distinguish two cases.
    
    If~$w$ has degree greater than two or if $w \in S$,
    then let~$G''$ be the graph obtained from~$G'$ by removing~$v$.
    Then $S'=S\setminus\{v\}$ is a vertex cover of~$G''$ and each leaf of~$G''$ is in~$S'$.
    By the inductive hypothesis, the number of vertices of~$G''$ is at most $2(k-1)-1$,
    thus the number of vertices of $G'$ is at most $2(k-1)-1+1 = 2k-2 < 2k-1$.
    
    Otherwise $w \notin S$ and the degree of~$w$ is at most two.
    In fact, the degree of~$w$ is two because if it were one,
    then $w$ would be a leaf that is in~$S$ and we would be in the first case.
    Hence, let $u\neq v$ be the other neighbor of~$w$.
    Observe that $u \in S$.
    Let~$G''$ be the graph obtained from~$G'$ by removing~$v$ and~$w$.
    Then, $S'=S\setminus\{v\}$ is a vertex cover of $G''$ and each leaf of~$G''$ is in~$S'$.
    By the inductive hypothesis, the number of vertices of~$G''$ is at most $2(k-1)-1$,
    so the number of vertices of~$G'$ is at most $2(k-1)-1+2 = 2k-1$. \qed
    \end{proof}

\section{Appendix: Missing Proofs for Cacti}\label{app:cacti}

In this section, we provide the missing proof details for \cref{thm:cactusGeneral-0} and \cref{thm:extensionToCacti}.

\subsection{Linear Cactus Supports}\label{sec:linear}

In this section, we consider linear path-based cactus supports. We show that the curve complexity is zero, the planar curve complexity is at most one, and it can be decided in linear time whether the planar curve complexity is zero or one. This result is used in the proof of \cref{thm:cactusGeneral-0} where we show how to decide in near-linear time whether the (planar) curve complexity of a path-based cactus support is zero.

\begin{theorem}
	\label{thm:cactusLinear}
	The curve complexity of 
	a linear path-based cactus support is 0.
\end{theorem}
\begin{proof}%
	Let $G$ be a linear path-based cactus support.
	Consider bridges as cycles of length two. We show by induction on the number of cycles the following slightly stronger statement:
	there is a (not necessarily planar) drawing of~$G$ in which all cycles are drawn convex and two incident edges are aligned if and only if they are contained in the same path.
	A single edge can clearly be drawn. Since we deal with a support of a linear hypergraph, every cycle that is not a bridge 
	contains edges of at least three different paths. Thus, if $G$ consists of a single cycle, we can draw it as a convex polygon with strictly convex angles at the vertices where two adjacent edges are contained in distinct paths and an angle of $\pi$ otherwise. 
	
	If $G$ has several cycles, let $C$ be a cycle of $G$ that contains exactly one cut vertex $v$.
	By the inductive hypothesis, $G-C+v$ has a drawing with the desired properties.
	Assume first that $C$ is a single edge $e$. If there is a path that contains $e$ and an edge $e'$ incident to $e$, we append $e$ to $v$ such that $e$ aligns with $e'$. This is possible, since the hypergraph is linear and, thus, $e'$ is not contained in any other path and, hence, is not aligned with any other edge.
	
	If $C$ is a real cycle let $e_1$ and $e_2$ be the two edges of $C$ incident to $v$. Again we align $e_1$ and $e_2$ with the respective edges of $G-C+v$  or with each other if there is a path containing $e_1$ or $e_2$ and the respective other edge.
	Analogously to the base case, we can close the cycle. 
	However, the alignment of $e_1$ and $e_2$ might imply that $e_1$ and $e_2$ have to go to different sides of a cycle that was already drawn. This will cause crossings when closing~the~cycle~$C$. \qed
\end{proof}

\begin{figure}[h]
	\centering
	\begin{minipage}{0.6\linewidth}
		\centering
		\includegraphics[page=1]{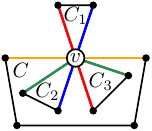}\hfil
		\includegraphics[page=2]{img/conflictG.pdf}\\[1ex]
            \includegraphics[page=3]{img/conflictG.pdf}
		\subcaption{planar curve complexity 1, curve complexity 0\label{fig:cactusPlanarCCgreater0}}
	\end{minipage}  \hfil
	\begin{minipage}{0.35\linewidth}
		\centering
            \includegraphics[page=4]{img/conflictG.pdf}\\[1ex]
		\includegraphics[page=5]{img/conflictG.pdf}
		\subcaption{planar curve complexity 0\label{fig:cactusPlanarCC0}}
	\end{minipage}
	\caption{Linear path-based cactus supports. Colors (except black) indicate paths. }
	\label{fig:badLinear}
\end{figure}

In a path-based support, we say that a cycle is \emph{aligned at a vertex $v$} if its two edges incident to~$v$ belong to the same path. For vertex $v$ of a path-based cactus support, consider the  \emph{constraint graph} $\mathcal H(v)$ that
has a node for each cycle containing $v$ and
has an edge between two nodes representing cycles $C$ and $C'$
if there is a path containing an edge in both~$C$ and~$C'$.
E.g., in \cref{fig:cactusPlanarCCgreater0}, $\mathcal H(v)$ contains the triangle $\langle C_1,C_2,C_3\rangle$, while in \cref{fig:cactusPlanarCC0} it contains a 4-cycle.

\begin{theorem}
	\label{thm:cactusLinearPlanar}
	The planar curve complexity of 
	a linear path-based cactus support $(G = (V, E), \p)$ is 0 if and only if for any vertex $v \in V$ (i)~at most one cycle of~$G$ is aligned at $v$ and (ii)~if a cycle
	is aligned at~$v$, then the constraint graph $\mathcal H(v)$ is bipartite.
	This can be determined in linear time.
\end{theorem}    

\begin{figure}[t]
	\centering
	\begin{minipage}[t]{0.3\textwidth}
		\centering
		\includegraphics[scale=0.9,page=4]{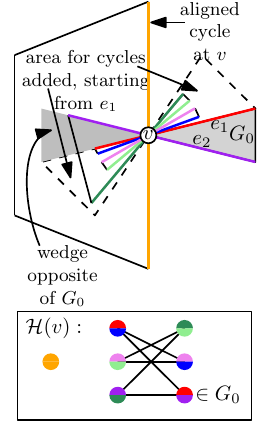}\hfil
		\subcaption{\label{fig:evenAv}even cycle in $\mathcal H(v)$}
	\end{minipage}
	\hfill
	\begin{minipage}[t]{0.32\textwidth}
		\centering
		\includegraphics[scale=0.9,page=5]{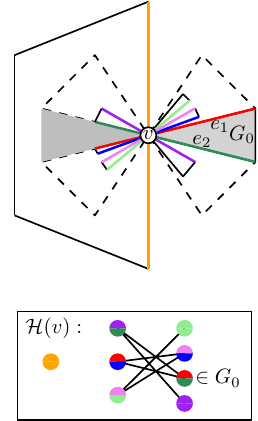}\hfil
		\subcaption{\label{fig:2pathsAv} path in $\mathcal H(v)$ via cycle in~$G_0$}
	\end{minipage}
	\hfill
	\begin{minipage}[t]{0.3\textwidth}
		\centering
		\includegraphics[scale=0.9,page=7]{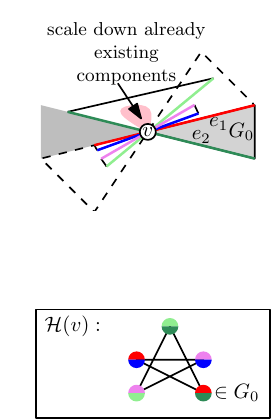}\hfil
		\subcaption{\label{fig:oddAv} odd cycle in $\mathcal H(v)$}
	\end{minipage}

	\caption{Constructing planar drawings with no bending paths of linear path-based cactus~supports.}
	\label{fig:cactus-construction}
\end{figure}

\begin{proof}%
	
	Two cycles that are both aligned at the same vertex
	must cross in a drawing without bends. If a cycle $C$ is aligned at a vertex~$v$, then the cycles of any path in $\mathcal H(v)$ must alternate between the interior and the exterior of $C$ in a planar drawing without bending~paths, implying that $\mathcal H(v)$ is bipartite. 
	Assume now that~$G$ is a linear path-based cactus support such that     
	Conditions (i)--(ii) are fulfilled for all vertices  of~$G$.  
	
	By induction on the number of cut vertices, we show that there is a planar drawing in which all cycles are drawn convex and two incident edges are aligned if and only if they are contained in the same path. 
	If $G$ consists of a single edge, it can be drawn in an arbitrary way. Since the hypergraph is linear, every cycle contains edges of at least three paths and, thus, can be drawn with three bends and with the desired properties.
	
	Now assume there are cut vertices.
	A \emph{component} of a cut vertex $v$ is the subgraph of~$G$ induced by a connected component of $G-v$ together with $v$.
	Choose a cut vertex $v$ such that all but at most one component $G_0$ of $v$ consist of a simple cycle or an edge.  By the inductive hypothesis, $G_0$ has a planar drawing without bends. 
	Let $C$ be the cycle (or bridge) of $G_0$ incident to $v$. We add the other components $G_1,\dots,G_\ell$
	in a small area around $v$; see \cref{fig:cactus-construction}. 
	
	If $C$ is not aligned at $v$, $G_1,\dots,G_\ell$ are added in the exterior of $C$. If one among $G_1,\dots,G_\ell$, say $G_1$, is aligned at $v$, we start with $G_1$.
	Then, by Condition~(i), cycle~$C$ cannot be aligned at $v$ and, thus, by our induction invariant, the exterior angle at the two edges of~$C$ being incident to $v$ is greater than $\pi$.
	Thus, we can draw $G_1$ in the exterior of $C$ in the desired way. 
	
	For a component $G_i$, $i=0,\dots,\ell$, let $e_1(G_i)$ and $e_2(G_i)$ be the (at most) two edges of $G_i$ incident to~$v$.
	Assume we have already drawn $G_i$.
	Then the lines through $e_1(G_i)$ and $e_2(G_i)$ divide the plane into four wedges $W_1,\dots,W_4$ around $v$, one of which, say $W_1$, contains $G_i$ (or $C$ if $i=0$).
	The opposite wedge of $G_i$ is then~$W_3$. E.g., the gray shaded area in \cref{fig:cactus-construction} shows the wedge between the lines through $e_1(G_i)$ and $e_2(G_i)$ that contains $C$ and the opposite wedge of $G_0$.
	Starting from $C$, we process $G_1,\dots,G_\ell$ in the order in which they appear on maximal paths in $\mathcal H(v)$. In the beginning and each time we have processed a maximal path, 
	we keep the invariant that at most one of the already drawn components has unfinished paths, i.e., contains edges that are contained in paths that also share an edge with a not yet drawn component.
	Moreover, if such a component $G'$ exists, then its opposite wedge is empty in some small vicinity around~$v$.
	If no component with unfinished paths exists, we add an arbitrary component in the exterior of any other component such that the opposite wedge is empty in some vicinity of $v$.
	This is always possible if we make the angle between the up-to-two edges incident to $v$ very small.
	So assume now there is a drawn component $G_i$ with unfinished paths.
	Starting from $G_i$ and following a maximal path $P$ in $\mathcal H(v)$, we keep drawing components into a small region next to the opposite wedge of $G_i$ and a small region next to the edge $e_1$ of $G_i$ that is contained in a path sharing an edge with the next component in $P$.
	The path $P$ ends if we either reach a leaf in $\mathcal H(v)$ or again component~$G_i$.
	\cref{fig:evenAv} illustrates the second case.
	In this case we close the last cycle using also the opposite wedge of $G_i$.
	In the case that we traversed an odd length cycle in $\mathcal H(v)$, we make sure that we do not enclose the component $G_0$ (by choosing the direction in which we traverse the cycle in $\mathcal H(v)$) and scale down the already drawn components that we enclose; see \cref{fig:oddAv}.
	Observe that this only happens if no cycle is aligned at $v$.
	If we ended at a leaf, and $G_i$ has still unfinished paths then let $G_j$ be the neighbor of $G_i$ that has not yet been drawn.
	Repeat this procedure, proceeding along a maximal path in $\mathcal H(v)$ starting from $G_i,G_j$; see~\cref{fig:2pathsAv}. \qed
\end{proof}

\begin{corollary}\label{cor:cactusLinearPlanar1}
	The planar curve complexity of a linear path-based support is 0 if the underlying support graph is a tree and at most 1 if the underlying support graph is a cactus.
\end{corollary}
\begin{proof}$ $
	The conditions in \cref{thm:cactusLinearPlanar}  are trivially fulfilled for a tree, since there are no cycles. 
	For a cactus support, follow the construction in the proof of \cref{thm:cactusLinearPlanar}.
	When inserting new components, split all paths at $v$ unless the path contains an edge of $G_0$. This way the conditions in \cref{thm:cactusLinearPlanar} are fulfilled. Observe that we only split paths when we first encounter them. Since the splits correspond to the bends in the drawing, it follows that each path bends at most once. \qed
\end{proof}

\begin{corollary}\label{cor:cactusLinearPlanar2}
	Given a linear path-based cactus support $(G,\p)$, the planar curve complexity of $(G,\p)$ can be determined in linear time and a drawing of $G$ minimizing the sum of the number of bends in all paths in $\p$ can also be constructed in linear time.
\end{corollary}
\begin{proof}
	The planar curve complexity of $(G,\p)$ is at most 1 and it can be determined in linear time whether it is 0.
	In order to construct a drawing of $G$ minimizing the total number of bends, we do the following. Let $v$ be a vertex. Let $\p_v$ be the set of paths in $\p$ that contain $v$ as an internal vertex. Let  $\p'_v$ be the set of paths $p \in \p_v$ such that the two edges of $p$ incident to $v$ are contained in one cycle of $G$. 
	In order to achieve a planar drawing without bending paths,
	we split paths as few times as possible so that the conditions in \cref{thm:cactusLinearPlanar} are fulfilled. I.e., we split all but one path in $\p'_v$ at $v$ if $\mathcal H_v$  contains no odd cycles and otherwise we split all paths in $\p'_v$ at $v$.
	Each split corresponds to a bend. \qed
\end{proof}

\cactusGeneral*
\label{thm:cactusGeneral-0*}
\begin{proof}
    It remains to prove the running time of our algorithm.
    Note that, since $G$ is planar, $|E| \in \oh(|V|)$.
    
    First, we check if at some vertex there are two paths
    that split, i.e., they share one edge~$e$ incident to~$v$
    but they continue in different edges~$e_1$ and~$e_2$ incident to~$v$.
    Essentially, we need to verify that, at every vertex,
    all incident edges can be matched into pairs
    that have precisely the same set of paths on them
    when we ignore the paths that end in~$v$.
    (Note that~$v$ can have odd degree if there
    is an edge incident to~$v$ containing only
    paths ending in~$v$; such an edge will be ignored.)
    To this end, we use, at every vertex~$v$, a union-find data structure
    (augmented by counters) that groups the set of paths
    containing the same edge incident to $v$ (but not ending at $v$)
    into one set.
    In detail, we proceed as follows.
    
    To ignore the paths that end in~$v$,
    we think of each edge $uv$ as two half-edges;
    one incident to~$u$ and the other incident to~$v$.
    We start by storing, for each half-edge~$e$, a list $\p(e)$ that contains
    all paths that traverse~$e$ and do not end in the incident vertex of~$e$.
    Further, we store for each vertex~$v$,
    a list $\p(v)$ that contains all paths having~$v$ as an inner vertex.
    These lists can be created in overall
    $\oh(\|\mathcal P\|)$ time.
    Then, for each vertex~$v$, we create a union-find data structure
    over the paths in $\p(v)$
    with the additional information that we store for
    each set how often we have accessed it.
    Initially each set contains a single path as one object and the access counter is zero.
    For each half-edge~$e$ incident to~$v$,
    we iterate over the paths in~$\p(e)$.
    If $\p(e)=\emptyset$, we skip~$e$.
    Otherwise, let $P_1$ be the first path in~$\p(e)$.
    Find $P_1$ in the union-find data structure
    and let its set be~$S$.
    
    Consider first the case that the access counter of~$S$ is~0. Then $P_1$ is the only path in $S$ and $P_1$ has not yet been considered.
    We increase the access counter of $S$ to~1, and
    we check for each path~$P \in \p(e) \setminus \{P_1\}$
    if the  access counter of the set $S'$ containing $P$ is~$0$, too.
    If yes, we make a union operation on $S'$ and $S$; we call the new set again $S$ with counter one.
    If no, then $P$ has appeared in a half-edge incident to~$v$
    where $P_1$ has not appeared; we abort and return false.
    If we do not abort, then in the end $S$ contains exactly the paths of $\p(e)$ and $S$ will not be subject to a union operation later.
    We say that $S$ \emph{originates} from~$e$.
    
    Consider now the case that the access counter of~$S$ was already~1.
    Then $S$ is the set $\p(e')$ of paths containing some edge $e'$ that was considered before $e$.
    We increase the access counter to~2, and
    we find for each path~$P \in \p(e) \setminus \{P_1\}$
    its set~$S'$ in the union-find data structure.
    If $S' \ne S$, then $P$ has not appeared
    in a half-edge where~$P_1$ has appeared; we abort and return false.
    If we do not abort, then $S$ contains all paths of $\p(e)$.
    
    Finally, consider the case that the access counter of~$S$ was already~2.
    Then there is a path~$P$ in~$S$ that has appeared
    in two half-edges that we have seen before.
    while $P_1$ has appeared in at most one of these half-edges;
    we abort and return false.

    If we do not abort, then in the end each set in the union-find data structure originates for some half-edge~$e'$ incident to $v$.
    Consider the set of half-edges incident to $v$ where no set in the union-find data structure originates from.
    For such a half-edge~$e$ with $\p(e) \ne \emptyset$, there is a half-edge $e'$ that was processed before $e$
    such that $\p(e) \subseteq \p(e')$ and $\p(e')$ is one of the sets of paths in the union-find data structure
    (i.e., the set that originates from~$e'$).
    Assume now that $\emptyset \neq \p(e) \subsetneq \p(e')$, let $P \in \p(e') \setminus \p(e)$, and let $e'' \neq e$ be a half-edge in $P$ incident to $v$.
    If $e''$ was processed before $e$, then the counter of the set $S'$ containing $\p(e')$ was 2 at the time when we processed~$e$.
    Thus, the algorithm had returned false.
    Otherwise, the counter of $S$ was two at the time when we processed $e''$ and the algorithm would have returned false at this point.
    Summarizing, we get that, if the algorithm does not return false, then for two edges $e$ and $e'$ incident to $v$ either $\p(e)=\p(e')$ or $\p(e) \cap \p(e') =\emptyset$, which implies that no two paths split.
    
    Concerning the runtime of this first step: If $v$ appears in $k_v$ paths as an inner vertex, then we have $2k_v$ find operations and at most $k_v-1 \leq |\p| - 1$ union operations.
    Thus, the run time  for $v$ is in $\mathcal \oh(k_v \cdot \alpha(|\p|))$, where $\alpha$ is the inverse Ackermann function.
    Summing up over all vertices, the total run time for this first step is in $\mathcal \oh(\sum_{v \in V} k_v \cdot \alpha(|\p|)) \subseteq \mathcal \oh(\|\mathcal P\| \cdot \alpha(|\p|))$.
    
    Second, we merge all paths that share an edge into one path.
    To this end, we initialize a union-find data structure over all paths in~$\p$.
    For each edge~$e$, we iterate over all paths traversing~$e$.
    Note that we can create corresponding lists for all edges in overall $\oh(\|\p\|)$ time.
    Let $P_1$ be the first path in the list of~$e$.
    We find the set~$S$ of~$P_1$ in the union-find data structure.
    For each following path~$P$ (in the list of the same edge),
    we find its set~$S'$ and,
    if $S' \ne S$, we apply the union operation to~$S$ and~$S'$.
    In the end, each set of the union-find data structure
    corresponds to a ``new'' path.
    We have $\oh(|V|)$ new paths since every edge is contained in at most one path,
    and we have $\oh(\|\p\|)$ find and $\oh(|\p|)$ union operations.
    The total running time of this step is again in
    $\mathcal \oh(\|\p\| \cdot \alpha(|\p|))$.
    
    Third, we assure that each pair of ``new'' paths
    intersects in at most one vertex. 
    Any two paths that share two vertices do so on some cycle.
    Since no two paths share an edge, this is only possible if the two paths together cover the cycle.
    In order to check that, we label the edges by the path they are contained in.
    By the preprocessing, each edge is contained  in at most one path and, by our general assumption, each edge is contained in at least one path.
    Now we only have to check whether there is a cycle whose edges have exactly two distinct labels.
    This can be done in linear time. Recall that the input graph is a cactus and that the cycles are edge-disjoint. Moreover, all cycles can be found in linear time, e.g., by DFS. 
    
    Fourth, we apply the linear-time algorithm from \cref{thm:cactusLinearPlanar}.
    Summing up over all these steps,
    the total running time is in 
    $\oh(\|\p\| \cdot \alpha(|\p|))$, i.e., near-linear. \qed
\end{proof}

\subsection{Non-Linear Cactus Supports}

\cactusGeneral*
\label{thm:cactusGeneral-0*}
\begin{proof}
    It remains to prove the running time of our algorithm.
    Note that, since $G$ is planar, $|E| \in \oh(|V|)$.
    
    First, we check if at some vertex there are two paths
    that split, i.e., they share one edge~$e$ incident to~$v$
    but they continue in different edges~$e_1$ and~$e_2$ incident to~$v$.
    Essentially, we need to verify that, at every vertex,
    all incident edges can be matched into pairs
    that have precisely the same set of paths on them
    when we ignore the paths that end in~$v$.
    (Note that~$v$ can have odd degree if there
    is an edge incident to~$v$ containing only
    paths ending in~$v$; such an edge will be ignored.)
    To this end, we use, at every vertex~$v$, a union-find data structure
    (augmented by counters) that groups the set of paths
    containing the same edge incident to $v$ (but not ending at $v$)
    into one set.
    In detail, we proceed as follows.
    
    To ignore the paths that end in~$v$,
    we think of each edge $uv$ as two half-edges;
    one incident to~$u$ and the other incident to~$v$.
    We start by storing, for each half-edge~$e$, a list $\p(e)$ that contains
    all paths that traverse~$e$ and do not end in the incident vertex of~$e$.
    Further, we store for each vertex~$v$,
    a list $\p(v)$ that contains all paths having~$v$ as an inner vertex.
    These lists can be created in overall
    $\oh(\|\mathcal P\|)$ time.
    Then, for each vertex~$v$, we create a union-find data structure
    over the paths in $\p(v)$
    with the additional information that we store for
    each set how often we have accessed it.
    Initially each set contains a single path as one object and the access counter is zero.
    For each half-edge~$e$ incident to~$v$,
    we iterate over the paths in~$\p(e)$.
    If $\p(e)=\emptyset$, we skip~$e$.
    Otherwise, let $P_1$ be the first path in~$\p(e)$.
    Find $P_1$ in the union-find data structure
    and let its set be~$S$.
    
    Consider first the case that the access counter of~$S$ is~0. Then $P_1$ is the only path in $S$ and $P_1$ has not yet been considered.
    We increase the access counter of $S$ to~1, and
    we check for each path~$P \in \p(e) \setminus \{P_1\}$
    if the  access counter of the set $S'$ containing $P$ is~$0$, too.
    If yes, we make a union operation on $S'$ and $S$; we call the new set again $S$ with counter one.
    If no, then $P$ has appeared in a half-edge incident to~$v$
    where $P_1$ has not appeared; we abort and return false.
    If we do not abort, then in the end $S$ contains exactly the paths of $\p(e)$ and $S$ will not be subject to a union operation later.
    We say that $S$ \emph{originates} from~$e$.
    
    Consider now the case that the access counter of~$S$ was already~1.
    Then $S$ is the set $\p(e')$ of paths containing some edge $e'$ that was considered before $e$.
    We increase the access counter to~2, and
    we find for each path~$P \in \p(e) \setminus \{P_1\}$
    its set~$S'$ in the union-find data structure.
    If $S' \ne S$, then $P$ has not appeared
    in a half-edge where~$P_1$ has appeared; we abort and return false.
    If we do not abort, then $S$ contains all paths of $\p(e)$.
    
    Finally, consider the case that the access counter of~$S$ was already~2.
    Then there is a path~$P$ in~$S$ that has appeared
    in two half-edges that we have seen before.
    while $P_1$ has appeared in at most one of these half-edges;
    we abort and return false.

    If we do not abort, then in the end each set in the union-find data structure originates for some half-edge~$e'$ incident to $v$.
    Consider the set of half-edges incident to $v$ where no set in the union-find data structure originates from.
    For such a half-edge~$e$ with $\p(e) \ne \emptyset$, there is a half-edge $e'$ that was processed before $e$
    such that $\p(e) \subseteq \p(e')$ and $\p(e')$ is one of the sets of paths in the union-find data structure
    (i.e., the set that originates from~$e'$).
    Assume now that $\emptyset \neq \p(e) \subsetneq \p(e')$, let $P \in \p(e') \setminus \p(e)$, and let $e'' \neq e$ be a half-edge in $P$ incident to $v$.
    If $e''$ was processed before $e$, then the counter of the set $S'$ containing $\p(e')$ was 2 at the time when we processed~$e$.
    Thus, the algorithm had returned false.
    Otherwise, the counter of $S$ was two at the time when we processed $e''$ and the algorithm would have returned false at this point.
    Summarizing, we get that, if the algorithm does not return false, then for two edges $e$ and $e'$ incident to $v$ either $\p(e)=\p(e')$ or $\p(e) \cap \p(e') =\emptyset$, which implies that no two paths split.
    
    Concerning the runtime of this first step: If $v$ appears in $k_v$ paths as an inner vertex, then we have $2k_v$ find operations and at most $k_v-1 \leq |\p| - 1$ union operations.
    Thus, the run time  for $v$ is in $\mathcal \oh(k_v \cdot \alpha(|\p|))$, where $\alpha$ is the inverse Ackermann function.
    Summing up over all vertices, the total run time for this first step is in $\mathcal \oh(\sum_{v \in V} k_v \cdot \alpha(|\p|)) \subseteq \mathcal \oh(\|\mathcal P\| \cdot \alpha(|\p|))$.
    
    Second, we merge all paths that share an edge into one path.
    To this end, we initialize a union-find data structure over all paths in~$\p$.
    For each edge~$e$, we iterate over all paths traversing~$e$.
    Note that we can create corresponding lists for all edges in overall $\oh(\|\p\|)$ time.
    Let $P_1$ be the first path in the list of~$e$.
    We find the set~$S$ of~$P_1$ in the union-find data structure.
    For each following path~$P$ (in the list of the same edge),
    we find its set~$S'$ and,
    if $S' \ne S$, we apply the union operation to~$S$ and~$S'$.
    In the end, each set of the union-find data structure
    corresponds to a ``new'' path.
    We have $\oh(|V|)$ new paths since every edge is contained in at most one path,
    and we have $\oh(\|\p\|)$ find and $\oh(|\p|)$ union operations.
    The total running time of this step is again in
    $\mathcal \oh(\|\p\| \cdot \alpha(|\p|))$.
    
    Third, we assure that each pair of ``new'' paths
    intersects in at most one vertex. 
    Any two paths that share two vertices do so on some cycle.
    Since no two paths share an edge, this is only possible if the two paths together cover the cycle.
    In order to check that, we label the edges by the path they are contained in.
    By the preprocessing, each edge is contained  in at most one path and, by our general assumption, each edge is contained in at least one path.
    Now we only have to check whether there is a cycle whose edges have exactly two distinct labels.
    This can be done in linear time. Recall that the input graph is a cactus and that the cycles are edge-disjoint. Moreover, all cycles can be found in linear time, e.g., by DFS. 
    
    Fourth, we apply the linear-time algorithm from \cref{thm:cactusLinearPlanar}.
    Summing up over all these steps,
    the total running time is in 
    $\oh(\|\p\| \cdot \alpha(|\p|))$, i.e., near-linear. \qed
\end{proof}

\extensionToCacti*
    \label{thm:extensionToCacti*}

\begin{proof}
    We modify the dynamic programming approach from \cref{thm:treesGeneralFPT-pathspervertex} to allow for the input to be a cactus support.
    Root $G$ at an arbitrary vertex $r$.
    For each cycle $C$ of $G$, let $r_C$ be the vertex of $C$ that is closest to $r$.
    We call $r_C$ the root of $C$.
    Let $e_C$ be one of the two edges of $C$ incident to $r_C$ and let $t_C$ be the other end vertex of $e_C$.
    Let $G'$ be the tree obtained from $G$ by removing, for each cycle $C$, the edge $e_C$.
    We will execute (an adjusted version) of our dynamic program on~$G'$.
    The main differences are that we keep track of the number of bends
    in each cycle (every cycle needs at least three bends)
    and that we include the paths going over~$e_C$ when processing
    any vertex of~$C$.
    Note that if a vertex $v$ lies in multiple cycles,
    then it is the root of all of these cycles except for possibly one.
    If there exists such a cycle where $v$ is not the root, we denote it by~$C(v)$.

    When processing a vertex $v$ of $G'$,
    we consider as $G_v$ the subgraph of $G$
    induced by the vertices of $G'$ rooted at~$v$ and,
    as before, the parent $\pi_v$ of~$v$.
    In our records, we consider not only the
    paths going through~$v$ but also the
    paths through $e_{C(v)}$, that is, we consider the set of paths
    $\p'(v) := \{p \in \p \mid e_v \in p \lor e_{C(v)} \in p\}$.
    Hence, this requires maintaining $2k$~paths instead of~$k$ paths in a record.
    Furthermore, for every record $B(v)$,
    we maintain an additional entry~$d$.
    The entry $d \in \{0, 1, 2, \ge 3\}$ describes
    the number of bends in $C(v)$ made so far.
    We will use this information to assure that every cycle has at least three bends.
    
    When creating new records,
    we also take the entry~$d$ and each edge~$e_C$ into account
    where~$C$ is a cycle containing the currently considered vertex~$v$.
    Whenever we arrive at a new cycle, i.e., $v = t_{C(v)}$,
    we try all combinations of records and arrangements of incident edges as before,
    but we additionally include the edge $e_{C(v)}$.
    If we align $e_{C(v)}$ and $v\pi_v$, then we set $d = 0$,
    otherwise we set $d = 1$.
    Similarly, let $v$ be a vertex that is contained in a cycle
    where it is not the root and $v \ne t_{C(v)}$.
    Let $v'$ be the child of $v$ in $C(v)$ and let $d'$ be the number
    of bends in $C(v)$ taken from the record $B(v')$ that we use for combining.
    Then, we set $d = d'$ if we align $v'v$ and $v\pi_v$ or,
    otherwise, we set $d = d' + 1$.
    (If $d'$ is set to $\ge 3$, then $d$ is set to~$\ge 3$, too.)
    Now let $v$ be any vertex.
    For each cycle $C$ where $v = r_C$, let $v'$ be the child of $v$ in $C$.
    When aligning edges at~$v$, we include the edge $e_C$.
    Moreover, we combine only records of $v'$ where $d$ is set to $\ge 3$
    or where $d = 2$ if we do not align the edges $e_C$ and $vv'$.
    
    It remains to guarantee the geometric realizability of $G_v$.
    To this end, we do the following test before we accept any new record.
    We proceed similarly as in the proof of~\cref{thm:cactusGeneral-0}.
    Let $(G_v, \mathcal Q)$ be the instance where we split the paths of $\p'(v)$ at bends (which we get by the chosen alignments).
    Note that $\mathcal Q$ contains no pair of conflicting paths.%
    \footnote{We say that two paths $p_1$ and $p_2$ have a conflict if and only if there are three edges $e,e_1,e_2$ sharing a common end vertex such that $p_1$ contains $e$ and $e_1$ while $p_2$ contains $e$ and $e_2$.}
    Let $\mathcal Q'$ be the set of paths constructed from $\mathcal Q$ by merging paths whenever they share an edge.
    Observe that no two paths in $\mathcal Q'$ share two vertices, otherwise there would be a cycle with less than three bends and we would have rejected the record.
    If we aim for a planar drawing and the conditions of \cref{thm:cactusLinearPlanar} are not fulfilled for $(G_v,\mathcal Q')$, then we have to reject the record.
    (If we do not aim for a planar drawing,
    there is a realization due to \cref{thm:cactusLinear}.)
    Observe that we only need to check the conditions of \cref{thm:cactusLinearPlanar} at vertex $v$.
    This is a local condition that is still true for the vertices further down in the tree. \qed
\end{proof}

\section{Appendix: Missing Proof Details for Orthogonal Drawings}
\begin{figure}[t]
    \centering
    \begin{minipage}[b]{0.45\textwidth}
        \centering
         \includegraphics[page=3]{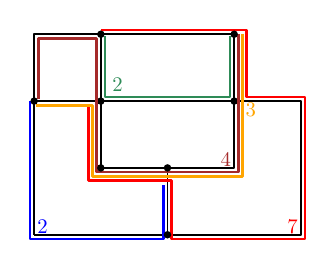}
    \subcaption{plane 4-graph support with 18 bends in total\label{fig:4plane-example}}
    \end{minipage}\hfil
    \begin{minipage}[b]{0.45\textwidth}
        \centering
         \includegraphics[page=5]{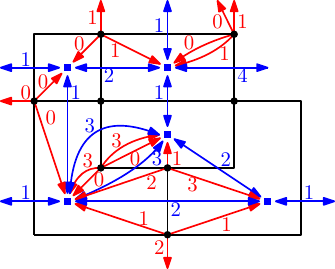}
    \subcaption{flow network (blue and red edges) with costs\label{fig:flow_network}}
    \end{minipage}
    
    \medskip
    
    \begin{minipage}[b]{\textwidth}
        \centering
         \includegraphics[page=7]{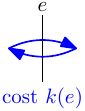} \hfil 
         \includegraphics[page=9]{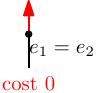} \hfil 
         \includegraphics[page=15]{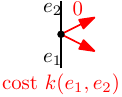} \hfil 
         \includegraphics[page=13]{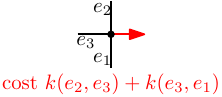} \hfil 
    \subcaption{costs in the flow network, $k(e)$ or $k(e_1,e_2)$ denote the number of paths containing $e$ or $e_1,e_2$, respectively\label{fig:costs}}
    \end{minipage}
    
    \caption{(a) A path-based plane 4-graph support (numbers indicate bends per path). (b, c) Flow model for minimizing the total number of bends in all paths.}
    \label{fig:4plane}
\end{figure}
\label{app:tamassia}
\tamassia*
\label{thm:tamassia*}
\begin{proof}
    Similar to the approach of Tamassia \cite{tamassia:87} for bend-minimum orthogonal drawings, we use a min-cost flow approach: The nodes, arcs, capacities, and demands are
    equal to the flow model of Tamassia.
    Only the costs are different.
    An example is given in \cref{fig:flow_network}.
    Each unit of flow represents a bend at an edge  or an angle of $\pi/2$ at a vertex beyond a first default angle of $\pi/2$.    
    The flow network $\mathcal N$ contains a node for each vertex and each face of $G$. The demand of every node representing a vertex $v$ of $G$ in $\mathcal N$ is $\deg(v)-4$, the demand of every node representing an inner face of $G$ is $\deg(f)-4$, and the demand of  the node representing the outer face of $G$ is $\deg(f)+4$. This models the fact that the angular sum around each vertex and the rotation around each inner face is $2\pi$, and the rotation of the outer face is $-2\pi$.
    There are the following two types of arcs;
    see
    \cref{fig:costs}.
    
    \paragraph{Face $\rightarrow$ Face Arcs.} For each edge $e$ of $G$ let $k(e)$ be the number of paths in $\p$ that contain $e$ and $f_1$ and $f_2$ be the faces incident to $e$ (potentially, $f_1=f_2$). 
        There is an arc from $f_1$ to $f_2$ and from $f_2$ to $f_1$ in $\mathcal N$ each with cost $k(e)$ and infinite capacity.  
        The amount of flow on these arcs model the number of convex bends on $e$ in $f_1$ and $f_2$,~respectively.
        
        \paragraph{Vertex $\rightarrow$ Face Arcs.}
        For each vertex $v$ and any two consecutive edges $e_1,e_2$ incident to~$v$, let $k(e_1,e_2)=k(e_2,e_1)$ be the number of paths in $\p$ that contain both $e_1$ and $e_2$. Assume that $e_1$ and $e_2$ are in counter-clockwise order and let 
        $f$ be the face  between $e_1$ and~$e_2$.
        We next describe the arcs in $\mathcal{N}$ between $v$ and $f$.
        A flow of $\phi$ units on these arcs models an angle of $(\phi+1)\pi/2$ between $e_1$ and~$e_2$ in $f$.
        \begin{itemize}
            \item 
            If the degree of $v$ is four, then there is no arc incident to $v$ in $\mathcal{N}$.
            \item 
            If the degree of $v$ is three, then there is one arc $e(e_1,e_2)$ from $v$ to $f$ in $\mathcal{N}$.
            Let $e_3$ be the third edge incident to~$v$.
            The angle between $e_1$ and $e_2$ is at most $\pi$.
            Hence, we set the capacity of $e(e_1,e_2)$ to~1.
            Moreover, if there is one unit of flow on $e(e_1,e_2)$ then $e_1$ and~$e_2$ are aligned
            implying that neither $e_1,e_3$ nor $e_2,e_3$ can be aligned.
            This is modeled by setting the cost of $e(e_1,e_2)$ to $k(e_1,e_3)+k(e_2,e_3)$.
            \item 
            If the degree of $v$ is two, then there are two arcs from $v$ to $f$ in $\mathcal N$,
            one with capacity~1 and cost~0 and the other with capacity~1 and cost~$k(e_1,e_2)$.
            In a cost-optimum drawing, the edge with positive cost would only be used
            if both units of flow out of $v$ go to the same side. In this case all paths over $v$ get a bend at $v$. 
            \item 
            If the degree of $v$ is one, then there is one arc $e(e_1,e_2)$ from $v$ to $f$ in $\mathcal{N}$,
            which has capacity~3 and cost~0. Note that there cannot be bends in degree-1 vertices.
        \end{itemize}
       
    Observe that the total cost of a flow equals the total number of bends on all paths except for bends in degree-four vertices. However, there is no choice anyway. So a min-cost flow yields a drawing with the minimum number of bends.
    Constructing $\mathcal N$ can be done in linear time
    and computing a min-cost flow in $\mathcal N$
    can be done in near-linear time~\cite{brand_etal:focs23}. \qed
\end{proof}

\end{document}